\documentclass[11pt]{article}

\usepackage{amssymb}
\usepackage{amsmath}
\usepackage{float}
\usepackage{xcolor}
\usepackage{graphicx}
\usepackage{subcaption} 
\usepackage{url}
\numberwithin{equation}{section}
\usepackage{amsthm}
\usepackage{caption}
\usepackage{ulem}
\usepackage{multirow}
\usepackage{array} 
\usepackage{booktabs}
\usepackage[numbers,sort&compress]{natbib}
\usepackage{algorithm}
\usepackage{algorithmicx}
\usepackage{algpseudocode}
\usepackage[shortlabels]{enumitem}
\usepackage{setspace}
\usepackage[a4paper,margin=1in]{geometry}
\usepackage{authblk}

\newtheorem{prop}{Proposition}

\newtheorem{defi}{Definition}
\newtheorem{remark}{Remark}

\newcommand{\prob}{\mathbb{P}}  
\newcommand{\tribu}{\mathcal{F}} 
\newcommand{\R}{\mathbb{R}}  
\newcommand{\domain}{\mathbb{X}}
\newcommand{\Esp}{\mathbb{E}}  
\newcommand{\Cov}{\mathbb{C}\mathrm{ov}}   

\newcommand{\Var}{\mathbb{V}\mathrm{ar}}  
\newcommand{\indexSet}{I}   
\newcommand{\exceptIndexSet}{- \indexSet}  

\newcommand{\Tr}[1]{\mathrm{Tr}\left[ #1 \right]}

\newcommand{\nvar}{d}   
\newcommand{\nOutputDimensions}{L}  
\newcommand{\indPoints}{i}

\newcommand{\nDoE}{n}  
\newcommand{\DoEpoint}[1]{x^{(#1)}}  
\newcommand{\DoEset}{\{\DoEpoint{1}, \dots, \DoEpoint{\nDoE} \}} 

\newcommand{\SobolNormalized}{S}  
\newcommand{\SobolUnnormalized}{D}  
\newcommand{\sobolTot}[1]{\overline{\SobolNormalized}_{#1}}
\newcommand{\sobolTotD}[1]{\overline{\SobolUnnormalized}_{#1}}
\newcommand{\sobolTotDmatrix}[1]{\overline{\mathbf{\SobolUnnormalized}}_{#1}}
\newcommand{\sobolClosed}[1]{\underline{\SobolNormalized}_{#1}}  
\newcommand{\sobolClosedD}[1]{\underline{\SobolUnnormalized}_{#1}}  

\newcommand{\sobolClosedDmatrix}[1]{\underline{\mathbf{\SobolUnnormalized}}_{#1}}

\newcommand{\GSI}{\underline{\mathrm{GSI}}}
\newcommand{\TotalGSI}[1]{\overline{\mathrm{GSI}}_{#1}}

\newcommand{\indOutputs}{\ell}
\newcommand{\inputVectorRandom}{X}  
\newcommand{\inputVectorDet}{x} 
\newcommand{\outputVectorRandom}{Y}  
\newcommand{\probDist}{\mu}    
\newcommand{\deterministicModel}{f}  
\newcommand{\inputPoint}{x}  
\newcommand{\DoE}{\mathcal{X}}   
\newcommand{\outputScalarDeterministic}{y}  
\newcommand{\outputVectorDeterministic}{\outputScalarDeterministic}  
\newcommand{\outputHighDimensional}{\outputScalarDeterministic_{\indOutputs}} 

\newcommand{\indPF}{k}

\newcommand{\PFSampleA}{\inputVectorRandom}  
\newcommand{\PFPointA}{\inputVectorRandom} 

\newcommand{\PFSampleB}{\inputVectorRandom'} 
\newcommand{\PFPointB}[1]{\inputVectorRandom^{* {#1}}} 

\newcommand{\pfEstim}[1]{{\widehat{#1}}^{\mathrm{pf}}}  
\newcommand{\pfEstimText}[2]{{\widehat{#1}}^{\mathrm{pf, #2}}}  
\newcommand{\nPF}{N}  
\newcommand{\pickFreezePoint}{x_{\indPF}}  
\newcommand{\pickFreezePointTilde}{\tilde{x}_{\indPF}} 
\newcommand{\pickFreezeVector}{x} 
\newcommand{\pickFreezeVectorTilde}{\tilde{\pickFreezeVector}} 

\newcommand{\PFoutput}[1]{\outputVectorRandom^{#1}} 
\newcommand{\PFoutputVec}[1]{\outputVectorRandom^{#1}} 
\newcommand{\PFoutputStar}[1]{\outputVectorRandom^{\indexSet, #1}}

\newcommand{\PFoutputStarTotal}[1]{\outputVectorRandom^{\exceptIndexSet, #1}}

\newcommand{\PFoutputStarVec}[1]{\outputVectorRandom^{\indexSet, #1}}

\newcommand{\CovMatrix}{\mathbf{D}}

\newcommand{\indGPR}{j}
\newcommand{\GaussianProcess}{Z} 
\newcommand{\MeanFunction}{m}
\newcommand{\KernelFunction}{K}
\newcommand{\condGP}{\GaussianProcess_n}
\newcommand{\ConditionalMean}{\MeanFunction_c}
\newcommand{\ConditionalKernel}{\KernelFunction_c}
\newcommand{\condGPVector}{\GaussianProcess_n}
\newcommand{\ObservationVector}{\outputVectorDeterministic}
\newcommand{\KernelMatrix}{\KernelFunction(\DoE, \DoE)}
\newcommand{\KernelVector}[1]{\KernelFunction(\DoE, #1)}
\newcommand{\KernelVectorRow}[1]{\KernelFunction(#1, \DoE)}
\newcommand{\NumberGPite}{N_{Z}}
\newcommand{\nX}{N_X}

\newcommand{\indBoot}{b}

\newcommand{\nBoot}{N_X}

\newcommand{\indBasis}{q}
\newcommand{\coef}{c}   
\newcommand{\basis}{v}  
\newcommand{\basisVec}{\basis_{., \indOutputs}}   
\newcommand{\basisMatrix}{\mathbf{V}}
\newcommand{\nbasis}{n_{b}}  
\newcommand{\coefVector}{\coef}  
\newcommand{\GramMatrix}{\mathbf{G}}

\begin{document}

\title{Estimation and model errors in Gaussian-process-based Sensitivity Analysis of functional outputs}

\author[1,2]{Yuri Taglieri Sáo}
\author[1]{Olivier Roustant}
\author[2]{Geraldo de Freitas Maciel}

\affil[1]{Institut de Mathématiques de Toulouse, Université de Toulouse, INSA\\
135, Avenue de Rangueil, 31077 Toulouse, Occitanie, France}

\affil[2]{Engineering College of Ilha Solteira, Civil Engineering Department, São Paulo State University
``Júlio de Mesquita Filho'' (UNESP)\\
Av. Brasil, 56, 15385-000 Ilha Solteira, São Paulo, Brazil}

\date{\today}

\maketitle

\begin{abstract}
Global sensitivity analysis (GSA) of functional-output models is usually performed by combining statistical techniques, such as basis expansions, metamodeling and sampling-based estimation of sensitivity indices. By neglecting truncation error from basis expansion, two main sources of errors propagate to the final sensitivity indices: the metamodeling-related error and the sampling-based, or \textit{pick-freeze} (PF), estimation error. This work provides an efficient algorithm to estimate these errors in the frame of Gaussian processes (GP), based on the approach of \citet{legratiet2014}.
The proposed algorithm takes advantage of the fact that the number of basis coefficients of expanded model outputs is significantly smaller than output dimensions. Basis coefficients are fitted by GP models and multiple conditional GP trajectories are sampled. Then, vector-valued PF estimation is used to speed-up the estimation of Sobol' indices and generalized sensitivity indices (GSI).
We illustrate the methodology on an analytical test case and on an application in non-Newtonian hydraulics, modelling an idealized dam-break flow. Numerical tests show an improvement of $15$ times in the computational time when compared to the application of \citet{legratiet2014} algorithm separately over each output dimension.
\end{abstract}


\section{Introduction}

Global sensitivity analysis (GSA) is a powerful set of techniques to estimate the influence of input variables over the variance of model outputs through sensitivity indices \citep{Saltelli2002, Iooss2015, DaVeiga2021}. 
For most practical applications the analytical calculation of these indices is not possible, except for a few simple expressions or specific classes of models, e.g. polynomial chaos expansion (PCE) \citep{Sudret2008}. To address this issue, Monte Carlo or quasi Monte Carlo methods are used to estimate numerically the sensitivity indices. Here, we focus on the \textit{pick-freeze} (PF) methods \citep{Sobol1993, Jansen1999, Saltelli2008, Janon2014, Gamboa2016}. The accurate computation of sensitivity indices through PF methods are computationally expensive, since they require a large number of model outputs (order of thousands) to be evaluated. In this case, one may replace the original true model by a fast-to-evaluate metamodel, which is an approximation of the original model, allowing the application of PF methods \citep{Iooss2015, DaVeiga2021}. Some examples of types of metamodels are the Gaussian process regression \citep{Williams1995}, polynomial chaos expansion \citep{Sudret2008}, artificial neural networks \citep{zou2008}, among others.

In many cases, models produce functional outputs and the evaluation of sensitivity indices in this functional scenario may be relevant. Usually, a basis-expansion technique is employed in order to reduce the dimensionality of outputs, by expanding output data into a functional basis composed of basis coefficients and basis components \citep{Marrel2011, Nagel2020, Perrin2021}. The few most important basis components are selected and the respective coefficients are used to construct metamodels. Then, vector-valued PF estimation is employed to compute sensitivity indices of basis coefficients, which are used to compute sensitivity maps (SMs) of Sobol' indices \citep{Li2020, Jung2023, Sao2025}. Additionally, generalized sensitivity indices (GSI) can be estimated to summarize the global contribution of a given input variable over the entire functional domain \citep{Lamboni2011, Perrin2021}. For the continuity of this work, we refer to this combination of techniques, i.e. basis expansions, metamodeling and PF methods, as \textit{functional GSA}.

The application of functional GSA presents three sources of error that affect the final sensitivity indices: (i) the basis-expansion truncation error, (ii) the PF estimation error; and (iii) the metamodeling error. The first one is neglected in this work, since the application of basis-expansion technique assumes a near-perfect representation of output data; \citet{Li2020} presents a remark about propagating basis-expansion truncation errors into Sobol' indices. The second one introduces error due to the numerical estimation of Sobol' indices and a common approach to estimate this error is the bootstraping approach \citep{Archer1997}. The third one is caused by the approximation of the true model with surrogate models trained on limited datasets. 
Notice that we assume that the model is well-specified. 
Some strategies can be used, such as bootstraping the initial design of experiments (DoE) \citep{Storlie2009, Rohmer2016, Palar2023} or, if available, using error bounds from metamodels, e.g. the variance of Gaussian process regression (GPR) \citep{Janon2014_2}. A robust methodology was proposed by \citet{legratiet2014} through an algorithm, where a set of random trajectories of a GPR surrogate model was subjected to bootstraping and allowed to obtain overall and metamodel-only uncertainties. We note that the methodology is initially developed for scalar outputs and we refer to \citet{Wirthl2023, Gauchy2024, Liu2024} as examples of applications of this methodology.

In the functional-GSA context, \citet{legratiet2014} methodology can be applied straightforwardly on each scalar output dimension (e.g. \citep{zhao2021, ye2022}). However, functional outputs are high-dimensional in practice and the computational cost, which is high due to sampling multiple GP trajectories and bootstraping, is scaled significantly if each output dimension is treated separately. Since functional GSA uses a basis expansion technique, the error estimation can be computed more efficiently by taking advantage of the basis components and coefficients (\textit{basis-derived} approach), instead of computing it dimension-by-dimension (\textit{dimension-wise} approach), as shown in \citep{Sao2025}. In this context, this work proposes to adapt efficiently the methodology proposed by \citet{legratiet2014} for functional GSA and extend the error estimation to the GSI. In fact, the contributions of this work require additional techniques in comparison with \citet{legratiet2014} algorithm, such as basis expansions, vector-valued PF methods, and usage of other sensitivity indices, i.e. the GSI. To the best of our knowledge, this method has not been explored yet in literature and contributes to the functional-output GSA literature in terms of computational efficiency of error estimation.

The next section presents a background of variance-based GSA (Subsection \ref{subsec:gsa}), GPR surrogate modeling (Subsection \ref{subsec:GPR}) and \citet{legratiet2014} algorithm (Subsection \ref{subsec:legratiet_algorithm}). Then, Section \ref{sec:contribution} details the methods used to build the proposed algorithm, followed by illustrative examples in Section \ref{sec:results}. The conclusions close the work in Section \ref{sec:conclusion}.

\section{Background}
\label{sec:background}

\subsection{Global sensitivity analysis}
\label{subsec:gsa}

In this subsection, we present briefly an overview of the main concepts of global sensitivity analysis (GSA). Further developments can be found in \citet{DaVeiga2021}. We consider $\nvar$ input variables, viewed as random variables $X_i$ with probability distribution $\probDist_{\inputVectorRandom_i}$.     
The vector of input variables is denoted by $\inputVectorRandom = (\inputVectorRandom_1, \dots, \inputVectorRandom_{\nvar})$, with probability distribution $\probDist_{\inputVectorRandom}$. The input variables are assumed to be independent. We denote the corresponding model output as $\outputVectorRandom = \deterministicModel(\inputVectorRandom)$, assumed to be square-integrable and scalar. The Sobol-Hoeffding decomposition allows us to define variance-based sensitivity indices, also called Sobol' indices. 
Here, we will focus on closed Sobol' indices, given by Eq. \ref{eq:sobol_indices}. 
\begin{equation}
\label{eq:sobol_indices}
    \sobolClosedD{\indexSet} (\outputVectorRandom) = \Var( \Esp [\deterministicModel_\indexSet( \inputVectorRandom ) | \inputVectorRandom_\indexSet]) \ ; \  \sobolClosed{\indexSet}(\outputVectorRandom) =  \frac{\Var( \Esp [\deterministicModel_\indexSet( \inputVectorRandom ) | \inputVectorRandom_\indexSet])}{\Var(\deterministicModel(\inputVectorRandom))}
\end{equation}
In this definition, we denote the index set $\indexSet \in \mathcal{P}(\{ 1,\dots,\nvar \})$, where $\mathcal{P}(\{ 1,\dots,\nvar \})$ is the set of subsets $\{ 1,\dots, \nvar \}$. We have adopted the underline notation, used e.g. in  \citep{owen2013}, writing $\sobolClosed{\indexSet}$ for the closed Sobol' index and $\sobolClosedD{\indexSet}$ for the unnormalized version.
When the size of $\indexSet$ is one, closed Sobol' indices coincide with first-order indices. Furthermore, higher-order and total Sobol' indices can be deduced from them \citep{DaVeiga2021}. 
Thus, there is no loss of generality in considering closed Sobol' indices, and we will focus on PF schemes to estimate closed Sobol' indices. 

In a PF scheme, one considers two independent random vectors $\PFSampleA$ and $\PFSampleB$ drawn from  $\probDist_\inputVectorRandom$. The model output is evaluated twice: firstly, by computing $\outputVectorRandom= \deterministicModel(\PFSampleA)$; and secondly, by fixing (freezing) the coordinates of $\PFSampleA$ in $\indexSet$ and choosing (picking) the other coordinates (denoted ${(\cdot)}_{\exceptIndexSet}$) in $\PFSampleB$, leading to $\outputVectorRandom^{\indexSet} = \deterministicModel(\PFSampleA_{\indexSet}, \PFSampleB_{\exceptIndexSet})$. Then, the normalized closed Sobol' index is given by Eq. \ref{eq:sobol_pf_scheme}. 
\begin{equation}
\label{eq:sobol_pf_scheme}
    \sobolClosed{\indexSet}(\outputVectorRandom) = \frac{\Cov(\outputVectorRandom, \outputVectorRandom^{\indexSet})}{\Var ( \outputVectorRandom )} = \frac{ \Esp[ \outputVectorRandom \outputVectorRandom^{\indexSet} ] - \left( \Esp \left[ \frac{\outputVectorRandom + \outputVectorRandom^\indexSet}{2} \right] \right)^2 }{ \Esp \left[ \frac{(\outputVectorRandom)^2 + (\outputVectorRandom^\indexSet)^2}{2} \right] - \left( \Esp \left[ \frac{\outputVectorRandom + \outputVectorRandom^\indexSet}{2} \right] \right)^2}
\end{equation}
Let us now consider two independent samples\footnote{If $\probDist$ is some probability distribution, a sample of size $\nDoE$ -- or $\nDoE$-sample -- of $\probDist$ will denote a family of $\nDoE$ independent and identically distributed (i.i.d.) random variables with law $\probDist$. Using the common abuse of notations, depending on the context, a $\nDoE$-sample may also denote the $\nDoE$ real numbers obtained as a realization of a (random) $\nDoE$-sample.} $\PFPointA^{1}, \dots, \PFPointA^{\nPF}$ and $\PFPointB{1}, \dots, \PFPointB{\nPF}$ of $\probDist_{\inputVectorRandom}$. We call PF input samples the pair of $\nPF$-samples $\PFPointA^{\indPF}$ and $\left( \PFPointA_{\indexSet}^{\indPF} ,\PFPointB{\indPF}_{\exceptIndexSet} \right)$ ($\indPF=1,\dots,\nPF$). We call PF output samples the corresponding output values, denoted
\begin{equation} \label{eq:PF_output}
\PFoutput{\indPF} = \deterministicModel \left(\PFPointA^{\indPF}\right), \qquad \PFoutputStar{\indPF} = \deterministicModel\left(\PFPointA_\indexSet^{\indPF}, \PFPointB{\indPF}_{\exceptIndexSet}\right) \qquad (\indPF=1, \dots, \nPF).
\end{equation}

With these notations, Eq.~\ref{eq:sobol_pf_scheme} leads to the PF estimator presented in Definition~\ref{def:pf_scalar_valued}. Studied in \citep{Janon2014}, it has nice statistical properties: consistency, asymptotical normality and asymptotical efficiency.

\begin{defi}[Pick-freeze estimator of closed Sobol' indices for scalar-valued functions] 
\label{def:pf_scalar_valued}
The empirical estimator of the closed Sobol' index $\widehat{\sobolClosed{\indexSet}}^{pf}$  corresponding to Eq. \eqref{eq:sobol_pf_scheme}, called Janon-Monod PF estimator \citep{Monod2006, Janon2014}, is given by 

\begin{equation}
    \label{Eq.pick_freeze_scalar_janon}
    \pfEstim{\sobolClosed{\indexSet}} = 
    \frac{\pfEstim{\sobolClosedD{\indexSet}}}{\pfEstim{D}}
\end{equation}
with 
\begin{eqnarray}
    \pfEstim{\sobolClosedD{\indexSet}} &=&
\frac{1}{\nPF} \sum_{\indPF=1}^\nPF \PFoutput{\indPF} \PFoutputStar{\indPF} - \left( \pfEstim{f_0} \right)^2 \\
\pfEstim{D} &=& \frac{1}{\nPF} \sum_{\indPF=1}^{\nPF} \left[ \frac{(\PFoutput{\indPF})^2 + (\PFoutputStar{\indPF})^2}{2} \right]- \left( \pfEstim{f_0} \right)^2
\end{eqnarray}
and 
\begin{equation} \label{eq:PF_f0}
\pfEstim{f_0} = \frac{1}{\nPF} \sum_{\indPF=1}^{\nPF} \left[ \frac{\PFoutput{\indPF} + \PFoutputStar{\indPF}}{2} \right].
\end{equation}
\end{defi}

We refer to \citep{Jansen1999, Janon2014, Gamboa2016} for other estimators for the closed Sobol' index, presented with their statistical properties. The methodology presented in this paper may work with some of them, provided that they can be expressed as a quadratic form of 
$\PFoutput{1},\dots, \PFoutput{\nPF}, \PFoutputStar{1},\dots, \PFoutputStar{\nPF}$ (see Remark 3 from \cite{Sao2025}). 

When the function $\deterministicModel$ is vector-valued, the unnormalized closed Sobol' index of $\deterministicModel(\inputVectorRandom) = \outputVectorRandom$, denoted $\sobolClosedDmatrix{\indexSet}(\outputVectorRandom)$, is defined as the covariance matrix of the random vector $\Esp[\outputVectorRandom|\inputVectorRandom_\indexSet]$:
$$\sobolClosedDmatrix{\indexSet}(\outputVectorRandom) = \Cov(\Esp[\outputVectorRandom|\inputVectorRandom_\indexSet])$$

The PF estimators of unnormalized closed Sobol' index and overall variance are immediately extended for vector-valued function, following Definition \ref{def:pf_vector_valued}.

\begin{defi}[Pick-freeze estimator for vector-valued functions] 
\label{def:pf_vector_valued}
The PF estimator of the unnormalized closed Sobol' index $\sobolClosedDmatrix{\indexSet}$, the unnormalized total Sobol' index $\sobolTotDmatrix{\indexSet}$ and the covariance matrix $\CovMatrix$ of a vector-valued function, are defined by the matrices:
\begin{eqnarray}
\label{eq:vector_valued_pf}
\pfEstim{\sobolClosedDmatrix{\indexSet}} &=& \frac{1}{\nPF} \sum_{\indPF=1}^{\nPF} \PFoutputVec{\indPF} (\PFoutputStarVec{\indPF})^{\top} - \pfEstim{f_0}\left(\pfEstim{f_0}\right)^\top \\
\pfEstim{\CovMatrix} &=& \frac{1}{\nPF} \sum_{\indPF=1}^{\nPF} \left[ \frac{\PFoutputVec{\indPF} (\PFoutputVec{\indPF})^\top + \PFoutputStarVec{\indPF}(\PFoutputStarVec{\indPF})^\top}{2} \right] - \pfEstim{f_0}\left(\pfEstim{f_0}\right)^\top
\end{eqnarray}
with 
$$
\pfEstim{f_0} = \frac{1}{\nPF} \sum_{k=1}^{\nPF} \left[ \frac{\PFoutputVec{\indPF} + \PFoutputStarVec{\indPF}}{2} \right] 
$$
\end{defi}

In the context of sensitivity analysis of functional outputs, it is useful to define an estimate for the global influence of the inputs over the output domain. Let us consider the following simulator:
\begin{equation}
\begin{array}{@{}r c c c@{}}
\outputVectorDeterministic :  &
  \domain \subseteq \R^{\nvar}
& \longrightarrow &
  \R^{\nOutputDimensions} \\[2pt]
&
  \inputPoint
& \longmapsto &
  \outputVectorDeterministic(\inputPoint)
\end{array}
\end{equation}
where $\outputVectorDeterministic(\inputPoint)=(\outputHighDimensional(\inputPoint))_{\indOutputs = 1, \dots, \nOutputDimensions}$ is the vector output and $\nOutputDimensions$ is the number of output dimensions. 

Definition \ref{def:gsi} introduces the GSI \citep{Lamboni2011, Perrin2021}, which are concise scalar sensitivity indices that estimate the global influence of inputs over the entire output domain. We will consider here the closed GSI, denoted by $\GSI$, which is the multidimensional version of the closed Sobol' index.

\begin{defi}[Generalized sensitivity indices] 
\label{def:gsi}
The closed generalized sensitivity index of $\outputVectorDeterministic(\inputVectorRandom)$ with respect to $\inputVectorRandom_{\indexSet}$, where $\indexSet \subseteq \{ 1, \dots, \nvar \}$, is given by Eq. \ref{eq:gsi_definition}. 
\begin{equation}
\label{eq:gsi_definition}
    \GSI_{\indexSet} = \frac{\Tr{\Cov(\Esp[\outputVectorDeterministic(\inputVectorRandom)|\inputVectorRandom_\indexSet])}}{\Tr{\Cov(\outputVectorDeterministic(\inputVectorRandom))}}
\end{equation}
where $\Tr{\cdot}$ denotes the trace of a matrix.
\end{defi}

Following Definitions \ref{def:pf_vector_valued} and \ref{def:gsi}, we can write a PF estimator for the closed first-order and total GSI.

\begin{defi}[PF estimator of generalized sensitivity indices]
\label{def:PF_gsi}
The PF estimator of the closed GSI is written
    \begin{equation}
    \label{eq:gsi_estimator}
    \pfEstim{\GSI_\indexSet} = \frac{\Tr{\pfEstim{\sobolClosedDmatrix{\indexSet}}}}{\Tr{\pfEstim{\CovMatrix}}}.
    \end{equation}
\end{defi}

PF schemes often need hundreds or thousands of sample points, which may be computationally expensive if we assume a costly deterministic code. For this reason, a \textit{metamodel} (also known as \textit{surrogate model}) can be used to generate accurate approximate data with substantially less cost. Although any metamodeling technique can be chosen to predict output values from PF input samples, this work explores the Gaussian process regression (GPR) technique, since it can provide the variance, or uncertainty, around the mean prediction.

\subsection{Gaussian process regression}
\label{subsec:GPR}

In this work, we focus on the GPR technique, 
which provides a probabilistic model to interpolate data. More information can be found in \citet{Williams1995}. Let $\deterministicModel : \domain \subseteq \R^{\nvar} \rightarrow \R$ be a multivariate function, modeling a deterministic computer code. As a prior knowledge, let us assume that $\deterministicModel$ is one trajectory of a Gaussian process $\GaussianProcess$ on $\domain$ with mean function $\MeanFunction$ and covariance function, or kernel, $\KernelFunction$. We will simply denote $\GaussianProcess \sim \mathcal{GP}(\MeanFunction, \KernelFunction)$.
Now, a posterior can be obtained from a design of experiments 
$\DoE = \DoEset$ (with $\DoEpoint{\indPoints} \in \domain$, $\indPoints=1, \dots, \nDoE$) and associated observations $\outputScalarDeterministic_\indPoints = \deterministicModel(\DoEpoint{\indPoints})$ ($\indPoints=1, \dots, \nDoE$). 
Indeed, from the properties of Gaussian vectors, $\GaussianProcess$ conditional on $\{\GaussianProcess(\DoEpoint{\indPoints}) = \outputScalarDeterministic_\indPoints, \indPoints=1, \dots, n\}$ is still a Gaussian process. If we denote $\condGP$ this conditional GP, we have
\begin{equation}
    \condGP \sim \mathcal{GP}(\ConditionalMean, \ConditionalKernel)
\end{equation}
where $\ConditionalMean$ and $\ConditionalKernel$ are given in closed-form by
\begin{equation}
\begin{aligned}
\ConditionalMean(\inputPoint) &= \MeanFunction(\inputPoint) + \KernelVectorRow{\inputPoint} \KernelMatrix^{-1} \ObservationVector  \qquad (\inputPoint \in \domain) \\
\ConditionalKernel(\inputPoint, \inputPoint') &= \KernelFunction(\inputPoint, \inputPoint') - \KernelVectorRow{\inputPoint} \KernelFunction(\DoE, \DoE)^{-1} \KernelVector{\inputPoint'} \qquad (\inputPoint, \inputPoint' \in \domain)
\end{aligned}
\label{eq:kriging_mean_variance}
\end{equation}
where $\mathbf{\outputScalarDeterministic} = (\outputScalarDeterministic_\indPoints)_{\indPoints=1, \dots, \nDoE}$ is the column vector of observations, 
$\KernelMatrix = (\KernelFunction(\DoEpoint{\indPoints}, \DoEpoint{j}))_{1 \leq \indPoints,j \leq \nDoE}$ is the covariance matrix at design points, $\KernelVector{\inputPoint} = (\KernelFunction(\DoEpoint{\indPoints}, \inputPoint))_{1 \leq \indPoints \leq \nDoE}$ is the column vector of covariances between a new point $x$ and the design points, and $\KernelVectorRow{\inputPoint} = \KernelVector{\inputPoint}^\top$.

A single trajectory of $\condGP$ at a vector of locations $\mathbf{t} = (t_1, \dots, t_m)$ can be obtained by simulating a Gaussian vector with mean $(\ConditionalMean(t_i))_{i=1, \dots, m}$ and covariance matrix $(\ConditionalKernel(t_i, t_j))_{1 \leq i,j \leq m}$. To do so, the Cholesky decomposition is often used. Notice that there exist improved simulation schemes for specific kernels (typically associated to stationary GPs) and/or specific DoEs (see e.g. \cite{legratiet2014}). We will not consider them here, as our aim is to present a general methodology.

\subsection{Algorithm proposed by \citet{legratiet2014}}
\label{subsec:legratiet_algorithm}

Here, we assume that the conditional GP $\condGP$ has already been built, from an initial DoE $\DoE$ and the corresponding observations, as explained in Section~\ref{subsec:GPR}. In order to propagate the errors to the Sobol' indices in the scalar-output scenario, an intuitive and \textit{crude} approach is shown by Algorithm \ref{algo:crude}. First, $\sobolClosed{\indexSet}$ is defined on a trajectory of the GP $\condGP$. This is a random variable related to the probability space $(\Omega_Z, \tribu_Z, \prob_Z)$ associated to the GP. Then, for a realization $\sobolClosed{\indexSet}(\omega_Z)$, where $\omega_Z \in \Omega_Z$, its PF estimator $\pfEstim{\sobolClosed{\indexSet}(\omega_Z)}$ is a random variable related to the probability space $(\Omega_X, \tribu_X, \prob_X)$ associated to the inputs uncertainty. Assuming independence between the two sources of uncertainties, $\pfEstim{\sobolClosed{\indexSet}}$ is a random variable related to the product probability space $(\Omega_X \times \Omega_Z, \sigma(\tribu_X \times \tribu_Z), \prob_X \otimes \prob_Z)$. Its law can be estimated by drawing independently trajectories of the GP and PF input samples from the inputs distribution. 

\begin{algorithm}[h!]  
\caption{ - Crude algorithm.}
\begin{algorithmic}[1]
\State \textbf{Input:} $\condGP$ (conditional Gaussian process); 
$\probDist_{\inputVectorRandom}$ (distribution of input variables); $\nPF$ (size of input samples); $\nX$ (number of input samples); $\NumberGPite$ (number of GP trajectories).

\For{$\indBoot = 1$ to $\nX$}  \Comment{Input sampling loop}
    \State Draw independently from $\probDist_{\inputVectorRandom}$ two $\nPF$-samples $\PFSampleA, \PFSampleB \in \domain^\nPF$;
    \State Generate PF input samples: $\pickFreezeVector = \PFSampleA$, $\pickFreezeVectorTilde = (\PFSampleA_{\indexSet}, \PFSampleB_{\exceptIndexSet})$;
    
    \For{$\indGPR = 1$ to $\NumberGPite$} \Comment{GP sampling loop}
    \State Sample a trajectory of $\condGP$ at the $2\nPF$ locations contained in  $\pickFreezeVector$ and $\pickFreezeVectorTilde$; 

    \State Compute $\pfEstim{\sobolClosed{\indexSet,\indGPR,\indBoot}}$ using Definition \ref{def:pf_scalar_valued} with $\PFoutput{\indPF} = \condGP(\pickFreezePoint)$ and $\PFoutputStar{\indPF} = \condGP(\pickFreezePointTilde)$;

    \EndFor
\EndFor

\noindent
\Return $\left[\pfEstim{\sobolClosed{\indexSet,\indGPR,\indBoot}}\right]_{\substack{\indGPR = 1,..., \NumberGPite \\ \indBoot=1,...,\nX}}$
\end{algorithmic}
\label{algo:crude}
\end{algorithm}

The drawback with Algorithm~\ref{algo:crude} is its computational cost. In particular, for each input sample $\pickFreezeVector$ (and $\tilde{\pickFreezeVector}$), a trajectory of the conditional GP is drawn at the locations of $\pickFreezeVector$ (and $\Tilde{\pickFreezeVector}$). To reduce the cost of GP sampling, \citet{legratiet2014} propose to draw input samples by bootstrap, i.e. by resampling with replacement from a first input sample obtained from the input distribution. Using bootstrap is equivalent to sample from the empirical distribution instead of the theoretical one, and thus makes little difference when $\nPF$ is large. The advantage is that the same (conditional) GP trajectory can be used for all boostrapped input samples, as they correspond to a subset of the computed values for the first sample.
The improved Algorithm is presented in Algorithm \ref{algo:leGratiet}.
We can see that, compared to Algorithm \ref{algo:crude}, the sampling order is reversed. First, a trajectory of the conditional GP is sampled (Line 5) at the locations of the PF input samples $\pickFreezeVector, \tilde{\pickFreezeVector}$ (Line 3), from which a first estimation of the Sobol' index is computed (Line 6). Then, new samples are generated by bootstrap from this first sample (Line 8 and 9), without the need of GP sampling. 

\begin{algorithm}[h!]
\caption{ proposed by Le Gratiet et al. (2014).}
\begin{algorithmic}[1]
\State \textbf{Input:} 
$\condGP$ (conditional Gaussian process); $\probDist_{\inputVectorRandom}$ (distribution of input variables); $\nPF$ (size of input samples); $\nBoot$ (number of input samples); $\NumberGPite$ (number of GP trajectories).

\State Draw independently from $\probDist_{\inputVectorRandom}$ two $\nPF$-samples $\PFSampleA, \PFSampleB \in \domain^\nPF$;

\State Generate PF input samples: $\pickFreezeVector = \PFSampleA$, $\pickFreezeVectorTilde = (\PFSampleA_{\indexSet}, \PFSampleB_{\exceptIndexSet})$;

\For{$\indGPR = 1$ to $\NumberGPite$} \Comment{GP sampling loop}

    \State Sample a trajectory of $\condGP$ at the $2\nPF$ locations contained in  $\pickFreezeVector$ and $\pickFreezeVectorTilde$; 

    \State Compute $\pfEstim{\sobolClosed{\indexSet,\indGPR,1}}$ using Definition \ref{def:pf_scalar_valued} with $\PFoutput{\indPF} = \condGP(\pickFreezePoint)$ and $\PFoutputStar{\indPF} = \condGP(\pickFreezePointTilde)$;

        \For{$\indBoot = 2$ to $\nBoot$} \Comment{Input resampling loop}
            \State Draw $\indPF_1, \dots, \indPF_\nPF$ independently from  $\mathcal{U}\{1, \dots, \nPF\}$;

            \State Compute $\pfEstim{\sobolClosed{\indexSet,\indGPR,\indBoot}}$ replacing the PF output samples $\PFoutput{1}, \PFoutputStar{1}, \dots, \PFoutput{\nPF}, \PFoutputStar{\nPF}$ by the resampled ones $\PFoutput{\indPF_1}, \PFoutputStar{\indPF_1}, \dots, \PFoutput{\indPF_\nPF},  \PFoutputStar{\indPF_\nPF}$;

            \State Redo Line 6. Store the results in $\pfEstim{\sobolClosed{\indexSet,\indGPR,\indBoot}}(\outputHighDimensional)$;
             
        \EndFor
\EndFor

\noindent
\Return $\left[\pfEstim{\sobolClosed{\indexSet,\indGPR,\indBoot}}\right]_{\substack{\indGPR = 1,..., \NumberGPite \\ \indBoot=1,...,\nBoot}}$
\end{algorithmic}
\label{algo:leGratiet}
\end{algorithm}

Two types of distributions are obtained for a given set of dimensions $\indexSet$. The first type is a distribution accounting for the metamodeling-only error, indicated by the subscripted index $(\cdot)_{\indexSet,\indGPR,1}$, where all $\indGPR$-th trajectories are explored without performing bootstraping. The second type is a distribution accounting for the overall error (metamodeling and estimation), indicated by the subscripted index $(\cdot)_{\indexSet,\indGPR,b}$, containing all $\indBoot$-th bootstraped estimations of every $\indGPR$-th trajectory. These distribution allow to compute statistics related to each error source separately, i.e. metamodeling- and estimation-only errors.

Finally, we note that \citet{legratiet2014} further explore several improvements, such as strategies to sample conditional GP trajectories from unconditioned ones, and to choose optimal number of PF points. Here, we limit our work to the simplest version described above. We can notice that this procedure has been introduced for scalar outputs. The next section extends it to functional outputs, such as time series or spatial maps.

\section{Errors on sensitivity indices for functional outputs}
\label{sec:contribution}

In this section, we adapt Algorithm \ref{algo:leGratiet} to quantify metamodeling and PF-estimation errors in the functional-GSA context, e.g. \citep{Li2020, Nagel2020, Jung2023, Sao2025}. We show that errors from different sources, namely from the metamodeling procedure and from estimation of Sobol' indices, can be computed efficiently over all output dimensions, by taking advantage of basis expansion techniques and by using bootstrap to draw input samples. A general description of basis expansion techniques is given in Subsection \ref{subsec:basis_expansion}. We also present the estimation of errors associated to the GSI \citep{Lamboni2011, Perrin2021}. The estimation of sensitivity indices using basis expansion is described in Subsection \ref{subsec:estimation_sensitivity}. Finally, Subsection \ref{subsec:algorithm} details the proposed algorithm of this work.

\subsection{Basis expansion of functional outputs}
\label{subsec:basis_expansion}
 
In the context of functional-output GSA, we assume that the number of output dimensions $\nOutputDimensions$ is large. Functional basis decomposition can be used to project output data onto a lower-dimensional subspace of $\nbasis \ll \nOutputDimensions$ by using a general linear basis expansion. Common choices include PCA, wavelets, Fourier basis, and B-splines, among others. Let us consider that output data $\outputHighDimensional$ are expanded in a general linear functional basis of dimension $\nbasis$, where $(\cdot)_{\indOutputs}$ represents the index of each output dimension, as follows:   
\begin{equation}
    \label{Eq.basis_expansion_truncated}
    \outputHighDimensional (\inputPoint) \approx \sum_{\indBasis = 1}^{\nbasis} \coef_\indBasis(\inputPoint) \basis_{\indBasis, \indOutputs} 
\end{equation}
where $\coefVector(\inputPoint) = (\coef_1(\inputPoint), \dots, \coef_{\nbasis}(\inputPoint))$ 
is the vector of basis coefficients and $\basisVec = (\basis_{1,\indOutputs}, \dots, \basis_{\nbasis,\indOutputs})$, is the vector of basis components defined in every output dimension $\indOutputs$.
We set $\basisMatrix = (\basis_{\indBasis, \indOutputs})_{1 \leq \indBasis \leq \nbasis, \\ 1 \leq \indOutputs \leq \nOutputDimensions}$ the matrix whose lines contain the basis vectors. 
The associated Gram matrix, containing the scalar products between the basis vectors $\basis_{\indBasis, .}$ ($\indBasis = 1, \dots, \nbasis$), can be computed as $\GramMatrix = \basisMatrix \basisMatrix^{\top}$.

In theory, the summation encompasses all $\nOutputDimensions$ terms; however, in order to effectively perform dimension-reduction, the summation is truncated into the most important $\nbasis$ terms of the expansion. In this approach, we neglect the approximation error (when we  project the output onto the linear space spanned by the $\basisVec$ functions). It is possible to account for it, at least approximately in the PCA framework, e.g. by adding a Gaussian noise whose variance is equal to the proportion of variance left unexplained. \citet{Li2020} also presents a brief discussion about the errors from basis expansions.

\subsection{Estimation of closed sensitivity indices using basis expansions}
\label{subsec:estimation_sensitivity}

We can leverage basis expansion to obtain fast formulas to compute both PF estimators of sensitivity maps and GSI. These formulas rely on the matrix-valued PF estimator of the vector of basis coefficients.

\begin{prop}[Basis-derived pick-freeze estimation of sensitivity maps] \label{prop:basisDerived_SobolIndex}
    The PF estimator of the closed normalized Sobol' indices of output dimensions  $\pfEstim{\sobolClosed{\indexSet}}(\outputHighDimensional) = \pfEstim{\sobolClosedD{\indexSet}}(\outputHighDimensional) / \pfEstim{D}(\outputHighDimensional)$ can be computed with
\begin{equation} \label{eq:reprojection}
     \pfEstim{\sobolClosedD{\indexSet}}(\outputHighDimensional) = \basisVec^\top \, \pfEstim{\sobolClosedDmatrix{\indexSet}}(\coefVector) \, \basisVec, 
    \qquad \pfEstim{D}(\outputHighDimensional) = \basisVec^\top \, \pfEstim{\CovMatrix}(\coefVector) \, \basisVec
\end{equation}
\end{prop}

\begin{proof}
    These formulas are derived in \citep{Sao2025}. We also refer to \cite{Li2020} for a similar result presented when the basis expansion is  PCA.
\end{proof}

\begin{prop}[Basis-derived pick-freeze estimation of generalized sensitivity indices]
\label{prop:gsi_estimator_basis_derived}
The basis-derived form of Definition \ref{def:gsi} is given by:

\begin{equation}\label{eq:gsi_theoretical_basis_derived}
    \GSI_\indexSet(\outputVectorDeterministic) = \frac{\Tr{\sobolClosedDmatrix{\indexSet} (\coefVector) \, \GramMatrix}}{\Tr{\CovMatrix (\coefVector) \, \GramMatrix}}
\end{equation}

Furthermore, the PF estimator of $\GSI_\indexSet(\outputVectorDeterministic)$ can be rewritten in the same form: 
\begin{equation}
\label{eq:gsi_estimator_basis_derived}
\pfEstim{\GSI_\indexSet(\outputVectorDeterministic)} = \frac{\Tr{\pfEstim{\sobolClosedDmatrix{\indexSet}}(\coefVector) \, \GramMatrix}}{\Tr{\pfEstim{\CovMatrix} (\coefVector)\, \GramMatrix}}
\end{equation}
\end{prop}

\begin{proof} With the notations of Section~\ref{subsec:basis_expansion}, 
we can write
\begin{equation} 
\label{eq:basis_expansion_appendix}
    \outputVectorDeterministic(\inputPoint)  = \basisMatrix^\top \coefVector(\inputPoint).
\end{equation} 
Now consider the denominator of GSI, as defined in \eqref{eq:gsi_definition}. We have
$$ \Cov(\outputVectorDeterministic(\inputVectorRandom)) = \Cov(\basisMatrix^\top \coefVector(\inputVectorRandom))   = \basisMatrix^\top \Cov(\coefVector(\inputVectorRandom)) \basisMatrix. $$
Using the commutativity property of trace, we obtain 
$$ \Tr{\Cov(\outputVectorDeterministic(\inputVectorRandom))}  
= \Tr{ \Cov(\coefVector(\inputVectorRandom)) \basisMatrix \basisMatrix^\top} 
= \Tr{ \Cov(\coefVector(\inputVectorRandom)) \GramMatrix}.$$
An analogous computation for the numerator of \eqref{eq:gsi_definition} gives 
$$ \Tr{\Cov(\Esp[\outputVectorDeterministic(\inputVectorRandom)|\inputVectorRandom_\indexSet])} = 
\Tr{\Cov(\Esp[\coefVector(\inputVectorRandom)|\inputVectorRandom_\indexSet]) \, \GramMatrix},$$
which leads to Equation~\ref{eq:gsi_theoretical_basis_derived}.\\

Let us now prove Equation~\ref{eq:gsi_estimator_basis_derived}.
Recall the definition of the GSI index
$$ \pfEstim{\GSI_\indexSet}(\outputVectorDeterministic) = \frac{\Tr{\pfEstim{\sobolClosedDmatrix{\indexSet}}(\outputVectorDeterministic)}}{\Tr{\pfEstim{\CovMatrix}(\outputVectorDeterministic)}}.$$
Let us first consider the denominator. From 
Proposition~\ref{prop:basisDerived_SobolIndex}, 
\begin{eqnarray*}
   \Tr{\pfEstim{\CovMatrix}(\outputVectorDeterministic)} &=& \sum_{\indOutputs=1}^\nOutputDimensions \pfEstim{D}(\outputHighDimensional) = 
\sum_{\indOutputs=1}^\nOutputDimensions
\basisVec^\top \, \pfEstim{\CovMatrix}(\coefVector) \, \basisVec\\
&=& \Tr{\sum_{\indOutputs=1}^\nOutputDimensions
\basisVec^\top \, \pfEstim{\CovMatrix}(\coefVector) \, \basisVec } = \Tr{ \pfEstim{\CovMatrix}(\coefVector) \sum_{\indOutputs=1}^\nOutputDimensions
  \, \basisVec \basisVec^\top \,}
\end{eqnarray*}
Remarking that $\GramMatrix = \sum_{\indOutputs=1}^\nOutputDimensions
  \, \basisVec \basisVec^\top \, $, we deduce that
$$ \Tr{\pfEstim{\CovMatrix}(\outputVectorDeterministic)} = \Tr{\pfEstim{\CovMatrix}(\coefVector) \GramMatrix}.$$
A similar computation shows that the numerator of the PF estimator of GSI verifies
$$ \Tr{\pfEstim{\sobolClosedDmatrix{\indexSet}}(\outputVectorDeterministic)} = \Tr{\pfEstim{\sobolClosedDmatrix{\indexSet}}(\coefVector) \GramMatrix}, $$
which concludes the proof.
\end{proof}

\begin{remark}[Calculation of the trace of two matrices] The trace of two matrices can be computed as follows
    $$ \Tr{A B} = \sum_{\indBasis} (A B)_{\indBasis,\indBasis} = \sum_{\indBasis} \sum_{\indBasis'} A_{\indBasis, \indBasis'} B_{\indBasis, \indBasis'} = \sum_{\indBasis, \indBasis'} (A \odot B)_{\indBasis, \indBasis'}$$
where $A \odot B$ is the Hadamard (element-wise) product.
Thus the traces above can be computed using $2 \nbasis^2$ operations.
\end{remark}

\begin{remark}[Covariance estimation] \label{remark:covariance_estimation}
    A similar formula for the GSI is derived in \citep{Perrin2021} and a particular case for PCA is considered in \citep{Lamboni2011, Perrin2021}. A common approach is to estimate the covariance matrix $\pfEstim{\CovMatrix}$ once considering output data from the DoE. In this work, we propose the estimator given by Eq. \ref{eq:gsi_estimator}, where the covariance matrix $\pfEstim{\CovMatrix} (\coefVector)$ is computed at each trajectory and at each bootstrap repetition. Both approaches present negligible differences in the experiments shown in Section \ref{sec:results} (see Appendix \ref{app:nGPR_nBoots}). However, our goal is not to study the asymptotic properties of the estimator, such as consistency, asymptotical efficiency and normality. This question has been addressed in \cite{Gamboa_Janon_Klein_Lagnoux_FunctionalGSA} for the Janon-Monod estimator and could be extended for other ones using the techniques presented therein.
\end{remark}

\subsection{Extension to other pick-freeze estimators}

As already remarked in \citep{Sao2025} (Remark 3), the formulas of Proposition~\ref{prop:basisDerived_SobolIndex} are not limited to the specific PF estimator of Definition~\ref{def:pf_scalar_valued} nor to closed indices. They are immediately extended to cases where the PF estimator is a quadratic form of the output PF samples  $\PFoutput{1}, \PFoutputStar{1}, \dots, \PFoutput{\nPF}, \PFoutputStar{\nPF}$.

Let us illustrate this point on the total Sobol' index $\sobolTot{\indexSet}$, which represents the proportion of variance explained by the input variables in $\indexSet$ and their interactions. It can be written as (see e.g. \cite{DaVeiga2021})
\begin{equation}
\label{eq:total_sobol_indices}
    \sobolTot{\indexSet} = \frac{\sobolTotD{\indexSet}}{D} , \qquad \sobolTotD{\indexSet} = D - \sobolClosedD{\exceptIndexSet} = \frac{1}{2} \Var\left( Y - Y^{-\indexSet} \right)
\end{equation}
where $\exceptIndexSet = \{ 1,\dots,\nvar \} \setminus \indexSet$.
The corresponding total GSI is defined by
\begin{equation}
\label{eq:totalgsi_definition}
    \TotalGSI{\indexSet} = 1 - \GSI_{\exceptIndexSet}.
\end{equation}
Using the first equality in the definition of $\sobolTotD{\indexSet}$ in \eqref{eq:total_sobol_indices}, we can consider the `plug-in' PF estimators of $\sobolTot{\indexSet}$ and $\TotalGSI{\indexSet}$ obtained by plugging in \eqref{eq:total_sobol_indices}
 and \eqref{eq:totalgsi_definition} the PF estimators of the closed Sobol' indices and closed GSI:
\begin{equation}
    \pfEstimText{\sobolTot{\indexSet}}{plug-in} = 1 - \pfEstim{\sobolClosed{\exceptIndexSet}}, \qquad 
    \pfEstimText{\TotalGSI{\indexSet}}{plug-in} = 1 - \pfEstim{\GSI_{\exceptIndexSet}}
\end{equation}
Alternatively, using the second equality in the definition of $\sobolTotD{\indexSet}$ in \eqref{eq:total_sobol_indices}, we obtain the Jansen estimator \citep{Jansen1999}
\begin{equation}
\label{eq:pick_freeze_total_jansen}
\begin{aligned}
\pfEstim{\sobolTot{\indexSet}}
&= \frac{\pfEstim{\sobolTotD{\indexSet}}}{\pfEstim{D}}, \\[2pt]
\pfEstim{\sobolTotD{\indexSet}}
&= \frac{1}{2 \nPF}\sum_{\indPF=1}^{\nPF}
\left( \PFoutput{\indPF} - \PFoutputStarTotal{\indPF} \right)^2, \\[2pt]
\pfEstim{\sobolTotDmatrix{\indexSet}}
&= \frac{1}{2 \nPF}\sum_{\indPF=1}^{\nPF}
\left( \PFoutput{\indPF} - \PFoutputStarTotal{\indPF} \right)
\left( \PFoutput{\indPF} - \PFoutputStarTotal{\indPF} \right)^\top .
\end{aligned}
\end{equation}

where the third expression corresponds to the vector-valued case. Using $\pfEstim{\sobolTotD{\indexSet}}$ is recommended to estimate $\sobolTotD{\indexSet}$ as it is asymptotically efficient, non-negative and vanishes if the true total Sobol' index is zero \cite{Fruth_TotalInteraction}. We can see that it is a quadratic form of the PF input samples $\PFoutput{1}, \PFoutputStar{1}, \dots, \PFoutput{\nPF}, \PFoutputStar{\nPF}$. Thus, the formulas of Proposition~\ref{prop:basisDerived_SobolIndex} hold when we replace the closed PF estimator by the (Jansen) PF one.
Furthermore, similarly to Definition~\ref{def:PF_gsi}, the (Jansen) PF estimator of the total GSI can be defined as
    \begin{equation}
    \label{eq:total_gsi_estimator}
    \pfEstim{\TotalGSI{\indexSet}} = \frac{\Tr{\pfEstim{\sobolTotDmatrix{\indexSet}}}}{\Tr{\pfEstim{\CovMatrix}}}.
    \end{equation}
Following the proof of Proposition~\ref{prop:gsi_estimator_basis_derived}, we can see that formula \eqref{eq:gsi_estimator_basis_derived}
also holds when replacing the closed PF estimator by the (Jansen) PF one.
In summary, we have:
\begin{equation}
\pfEstim{\sobolTotD{\indexSet}}(\outputHighDimensional) = \basisVec^\top \, \pfEstim{\sobolTotDmatrix{\indexSet}}(\coefVector) \, \basisVec,
\qquad
\pfEstim{\overline{\text{GSI}_\indexSet} 
(\outputVectorDeterministic)} = \frac{\Tr{\pfEstim{\sobolTotDmatrix{\indexSet}}(\coefVector) \, \GramMatrix}}{\Tr{\pfEstim{\CovMatrix} (\coefVector)\, \GramMatrix}}.
\end{equation}

\subsection{Proposed algorithm}
\label{subsec:algorithm}
Here, we use Algorithm \ref{algo:leGratiet} for the basis components by employing a vector-valued PF estimation and then compute the Sobol' indices at the original output dimensions and the GSI. This results Algorithm \ref{proposed_algorithm}, which effectively adapts Algorithm \ref{algo:leGratiet} to the functional GSA context.\\

To start Algorithm \ref{proposed_algorithm}, we assume that the basis components $\basisVec$ and the Gram matrix $\GramMatrix$ have been computed from a basis expansion (Eq. \ref{Eq.basis_expansion_truncated}) with $\nbasis$ components. We also assume that independent GP models $\condGP^{(\indBasis)}$, with $\indBasis=1,...,\nbasis$, were fitted for all components of $\coefVector(\inputPoint)$. We denote by $\condGPVector$ the vector-valued GP, where for all $\inputPoint \in \domain$, $\condGPVector(\inputPoint) = (\condGP^{(\indBasis)}(\inputPoint))_{\indBasis = 1 , \dots, \nbasis}$ is a column vector containing all GP values at $\inputPoint$. These inputs are defined in Line 1 of the proposed algorithm.

Two samples $\PFSampleA$ and $\PFSampleB$ of size $\nPF$ are drawn independently from a probability distribution $\probDist_{\inputVectorRandom}$ in Line 2, allowing the construction of two PF input samples $\pickFreezeVector = \PFSampleA$ and $\pickFreezeVectorTilde = (\PFSampleA_\indexSet, \PFSampleB_{\exceptIndexSet})$ in Line 3. Then, a $\indGPR$-th trajectory of $\condGPVector$ is sampled at the locations of the PF input vector by sampling independently one trajectory of each GP $\condGP^{(\indBasis)}$ ($\indBasis=1, \dots, \nbasis$) (Line 5). A practical remark about sampling trajectories of GPs computationally is given in Remark \ref{remark:sampling_trajectories}. For this trajectory, we estimate through Definition \ref{def:pf_vector_valued} the vector-valued unnormalized closed Sobol' indices of the basis components $\pfEstim{\sobolClosedDmatrix{\indexSet,\indGPR,1}}\in \R^{\nbasis\times \nbasis}$ and the overall covariance matrix of basis components $\pfEstim{\CovMatrix}_{\indGPR,1}\in \R^{\nbasis\times \nbasis}$ (Line 6). These matrices allow the computation of closed Sobol' indices for each output dimension $\pfEstim{\sobolClosed{\indexSet,\indGPR,1}}(\outputHighDimensional) \in \R$ (Eq. \ref{eq:reprojection}) in Lines 7-9. Next, in Line 10, the GSI can be estimated empirically using Eq. \ref{eq:gsi_estimator_basis_derived} presented in Proposition \ref{prop:gsi_estimator_basis_derived}. Note that the subscripted index $(\cdot)_{\indexSet,\indGPR,1}$ refers to the estimations obtained without bootstrap ($\indBoot=1$) for a given $\indGPR$-th trajectory.

Still for the same $\indGPR$-th trajectory, a bootstrap approach is used to obtain a distribution for $\pfEstim{\sobolClosed{\indexSet}}(\outputHighDimensional)$ and $\pfEstim{\GSI_{\indexSet}}$.The estimations of the Sobol' index at each output dimension and the GSI can be computed in Line 13 without further GP sampling. This process is repeated for $\nBoot$ times (Lines 11-14) and notice that the subscripted index $(\cdot)_{\indexSet,\indGPR,\indBoot}$ refers to the estimations obtained by bootstraping ($\indBoot=1,\dots,\nBoot$) for a $\indGPR$-th trajectory.

Consequently, two distributions of estimated sensitivity indices are produced: the distribution of Sobol' indices of each output dimension $\pfEstim{\sobolClosed{\indexSet,\indGPR,\indBoot}}(\outputHighDimensional)$ with size $\left( \nOutputDimensions \times \NumberGPite \times \nBoot \right)$; and the distribution of generalized sensitivity indices $\pfEstim{\GSI_{\indexSet}}$ with size $\left( \NumberGPite \times \nBoot \right)$. Notice that for bootstrap index $\indBoot=1$ is related only to the metamodeling error, while the entire distributions are related to both metamodeling and PF estimation errors. The distributions allow us to obtain relevant statistics, such as median and percentiles (e.g. boxplots), and represent separately the variation due to the sources of errors, i.e. metamodeling and estimation.

\begin{algorithm}
\caption{ proposed by this work. }
\begin{algorithmic}[1]
\State \textbf{Input:} $\probDist_{\inputVectorRandom}$ (PDFs of input variables); $\nPF$ (size of PF input samples); $\NumberGPite$ (number of GPR trajectories); $\condGPVector$ (vector-valued GP); $\basisVec$ (basis components); $\GramMatrix$ (Gram matrix); $\nBoot$ (number of input samples).

\State Draw independently from $\probDist_{\inputVectorRandom}$ two $\nPF$-samples $\PFSampleA, \PFSampleB \in \domain^\nPF$;

\State Form PF input samples: $\pickFreezeVector =\PFSampleA$, $\pickFreezeVectorTilde = (\PFSampleA_{\indexSet}, \PFSampleB_{\exceptIndexSet})$;

\For{$\indGPR = 1$ to $\NumberGPite$} 
\Comment{GP trajectories loop}

    \State Sample a trajectory of $\condGPVector$ at the $2 \nPF$ locations of the vector $(\pickFreezeVector, \pickFreezeVectorTilde)$;

    \State Compute $\pfEstim{\sobolClosedDmatrix{\indexSet,\indGPR,1}}$ and $\pfEstim{\CovMatrix_{\indGPR,1}}$ using Definition \ref{def:pf_vector_valued} with $\PFoutputVec{\indPF} = \condGPVector(\pickFreezePoint)$ and $\PFoutputStarVec{\indPF} = \condGPVector(\pickFreezePointTilde)$;

    \For{$\indOutputs=1$ to $\nOutputDimensions$}
    \Comment{Output dimension loop}
        \State Compute $\pfEstim{\sobolClosed{\indexSet,\indGPR,1}}(\outputHighDimensional)$ using Eq. \ref{eq:reprojection};
    \EndFor

    \State{Compute $\pfEstim{\GSI_{\indexSet,\indGPR,1}}$ using Proposition \ref{prop:gsi_estimator_basis_derived}};

        \For{$\indBoot = 2$ to $\nBoot$} \Comment{Input resampling loop}

            \State Draw $\indPF_1, \dots, \indPF_\nPF$ independently from  $\mathcal{U}\{1, \dots, \nPF\}$;

            \State Compute $\pfEstim{\sobolClosedDmatrix{\indexSet,\indGPR,\indBoot}}$ and $\pfEstim{\CovMatrix_{\indGPR,\indBoot}}$ replacing the PF output samples $\PFoutputVec{1}, \PFoutputStarVec{1}, \dots, \PFoutputVec{\nPF}, \PFoutputStarVec{\nPF}$ by the resampled ones $\PFoutputVec{\indPF_1}, \PFoutputStarVec{\indPF_1}, \dots, \PFoutputVec{\indPF_\nPF}, \PFoutputStarVec{\indPF_\nPF}$;

            \State Redo Lines 7-10. Store the results in $\pfEstim{\sobolClosed{\indexSet,\indGPR,\indBoot}}(\outputHighDimensional)$ and $\pfEstim{\GSI_{\indexSet,\indGPR,\indBoot}}$. 
            
        \EndFor
\EndFor

\noindent 
\Return $\left[\pfEstim{\sobolClosed{\indexSet,\indGPR,\indBoot}}(\outputHighDimensional)\right]_{\substack{\indOutputs=1,\dots,\nOutputDimensions \\ \indGPR = 1,..., \NumberGPite \\ \indBoot=1,...,\nBoot}}$ and $\left[\pfEstim{\GSI_{\indexSet,j,\indBoot}}\right]_{\substack{\indGPR = 1,..., \NumberGPite \\ \indBoot=1,...,\nBoot}}$.

\end{algorithmic}
\label{proposed_algorithm}
\end{algorithm}

\begin{remark}[About practical implementation of GP sampling]
\label{remark:sampling_trajectories}
Considering computational implementation, GP trajectories can be generated in batches (outside the GP trajectories loop) or individually (within the GP trajectories loop). For the first option, a $\NumberGPite$-sized batch of trajectories of $\condGPVector$ is computed before Line 4 (start of GP trajectories loop). This approach is computationally faster because the Cholesky decomposition is performed only once. In contrast, the second option implies to perform Cholesky decomposition $\NumberGPite$ times, i.e. for each individual trajectory. The drawback of the ``batch'' option is that more computational storage is required, resulting in matrices of size $\nbasis \times \NumberGPite \times \nPF$. Both options are valid for the proposed algorithm.
\end{remark}

\begin{remark}[About computational cost]
    The standard approach to assess metamodeling and estimation errors in functional GSA is to apply \citet{legratiet2014} algorithm over all $\nOutputDimensions$ output dimensions \citep{zhao2021, ye2022}, i.e. a dimension-wise (DW) approach. Given pre-computed GP trajectories, the cost to estimate a single sensitivity map (SM) $\pfEstim{\sobolClosed{\indexSet}}(\outputHighDimensional)$ ($\indOutputs=1,\dots,\nOutputDimensions$) is $\mathrm{Cost_{DW}} = 4(\nbasis+2)\nPF\nOutputDimensions$, which is multiplied by $\NumberGPite$ and $\nBoot$ to obtain a distribution of SM. On the other hand, considering the novel basis-derived (BD) algorithm presented in this work, given the same pre-calculated GPR predictions, a single SM is computed with the cost of $\mathrm{Cost_{BD}} = 2\nbasis(3\nbasis+1)\nPF + 3\nbasis(\nbasis+1)\nOutputDimensions$, also scaled by $\NumberGPite$ and $\nBoot$ to obtain a distribution of SM. The relative cost between both methods is exactly the relative cost given by Proposition 3 of \citet{Sao2025} is lower-bounded by $H(2\nPF, \nOutputDimensions) / (3 \nbasis)$, where $H$ is the harmonic mean. It proves the efficiency of the novel algorithm over the standard approach to assess estimation and modeling errors.
   
\end{remark}

The next section shows the application of Algorithm \ref{proposed_algorithm} to an analytical case (Campbell2D function) and a data-driven application (idealized and gradual dam-break flow of non-Newtonian fluid).

\section{Applications}
\label{sec:results}

In this section, we explore two cases: first, the Campbell2D function (see e.g. \citep{Marrel2011, Perrin2021, Sao2025}), and the idealized and gradual dam-break flow of non-newtonian fluids \cite{Sao2025}. We use principal component analysis (PCA) as linear basis expansion technique. For the GPR, the standard Matérn 5/2 kernel is employed and the maximum likelihood estimator is used to estimate the hyperparameters. The accuracy of the metamodels' prediction is evaluated by the Nash-Sutcliffe efficiency, also called $Q^2$ metric, described in \citep{Marrel2011, Sao2025}. Regarding PF estimation, we employ the Janon-Monod estimator to calculate the first-order indices \citep{Janon2014} and the Jansen estimator to calculate the total indices \citep{Jansen1999}, because of their desirable statistical properties (asymptotical efficiency and consistency) \citep{Janon2014, DaVeiga2021}.

For both cases, we illustrate the ability of Algorithm \ref{proposed_algorithm} to produce distributions of sensitivity indices, representing metamodeling-only and overall (metamodeling and PF estimation) errors over Sobol' indices and GSI. The errors are represented by using boxplots\footnote{The boxplots here are defined by the median value (Q2), the interquartile (Q1 and Q3) range within the box, whiskers that span $1.5$ times the interquartile range and outlier points outside the whiskers.}. Two approaches may be used:

\begin{enumerate}
    \item The first approach aims to capture the metamodeling-only errors: given a DoE of size $\nDoE$, DoE subsets of different sizes are used in Algorithm \ref{proposed_algorithm} while keeping fixed the remaining parameters, i.e $\nPF, \nBoot,\NumberGPite$.

    \item The second approach aims to capture the overall errors: different sizes of PF samples $\nPF$ are used in Algorithm \ref{proposed_algorithm} while keeping fixed the remaining parameters, i.e $\nDoE, \nBoot,\NumberGPite$.
\end{enumerate}

The isolated PF-estimation error can be computed straightforwardly as the difference between both approaches. The parameters $\nBoot = 50$ and $\NumberGPite = 200$ are considered to be fixed and were subjected to tests for both cases. No significant difference was found by using more GP trajectories or more bootstrap repetitions, as shown in Appendix \ref{app:nGPR_nBoots}.

\subsection{Campbell2D function}
\label{subsec:campbell2d_results}

The Campbell2D function \citep{Marrel2011, Perrin2021, Sao2025} is originally defined as a function from $[-1, 5]^8$ to  $[-90,90]^2$. Here, we consider a discretized version where, for each input value $\inputVectorDet$, the output $f(\inputVectorDet)$ is a vector obtained by restricting the continuous values on a uniform grid of size $\nOutputDimensions = 64 \times 64$.

\begin{equation}
\label{eq:campbell2D}
\begin{array}{@{}r c c c@{}}
\deterministicModel: &
 \left[-1,5\right]^8
& \longrightarrow &
  \R^\nOutputDimensions \\[2pt]
&
  \inputVectorDet=\left(x_1,\ldots,x_8\right)
& \longmapsto &
  f(\inputVectorDet) = \left(\outputHighDimensional\!\left(\inputVectorDet\right) \right)_{\indOutputs = 1, \dots, \nOutputDimensions}
\end{array}
\end{equation}

A DoE of size $\nDoE = 200$ was generated by using the Latin-hypercube sampling (LHS) technique with uniform PDFs for the input variables. For all subsets of $\nDoE$, $\nbasis = 10$ was sufficient to explain more than $99\%$ of the variance. In general, the metamodels' predictions were acceptable, as indicated by the $Q^2$ maps shown in Appendix \ref{app:q2}, which displays the 5th, 50th and 95th percentiles of $Q^2$ from a sample containing $\NumberGPite=100$ GP trajectories. Since the predictions are not perfect (i.e. $Q^2 = 1)$, the metamodeling procedure is still a source of errors over the sensitivity indices. However, lower values of $Q^2$ are associated with more errors, which is reflected to the sensitivity indices, as observed in \citep{legratiet2014}. The metamodeling error is discussed in the sequel.

Figure \ref{fig:DOE_variation} shows the comparison of first-order GSI boxplots in function of $\nDoE$ subsets, for all input variables $(x_1, \dots, x_8)$. Here, we limit the presentation to first-order GSI; total GSI results are available in Appendix \ref{app:total_order}. By observing the metamodel-only data, it is noticeable that the coverage of the whiskers reduces as the size of DoE increases. It means that the improvement of predictive capability of metamodels makes the GSI to vary less, leading to less errors due to metamodeling. As the size of DoE increases, overall errors vary exclusively due to the reduction in metamodeling errors.

\begin{figure}[h]
    \centering
    \includegraphics[width=0.8\linewidth]{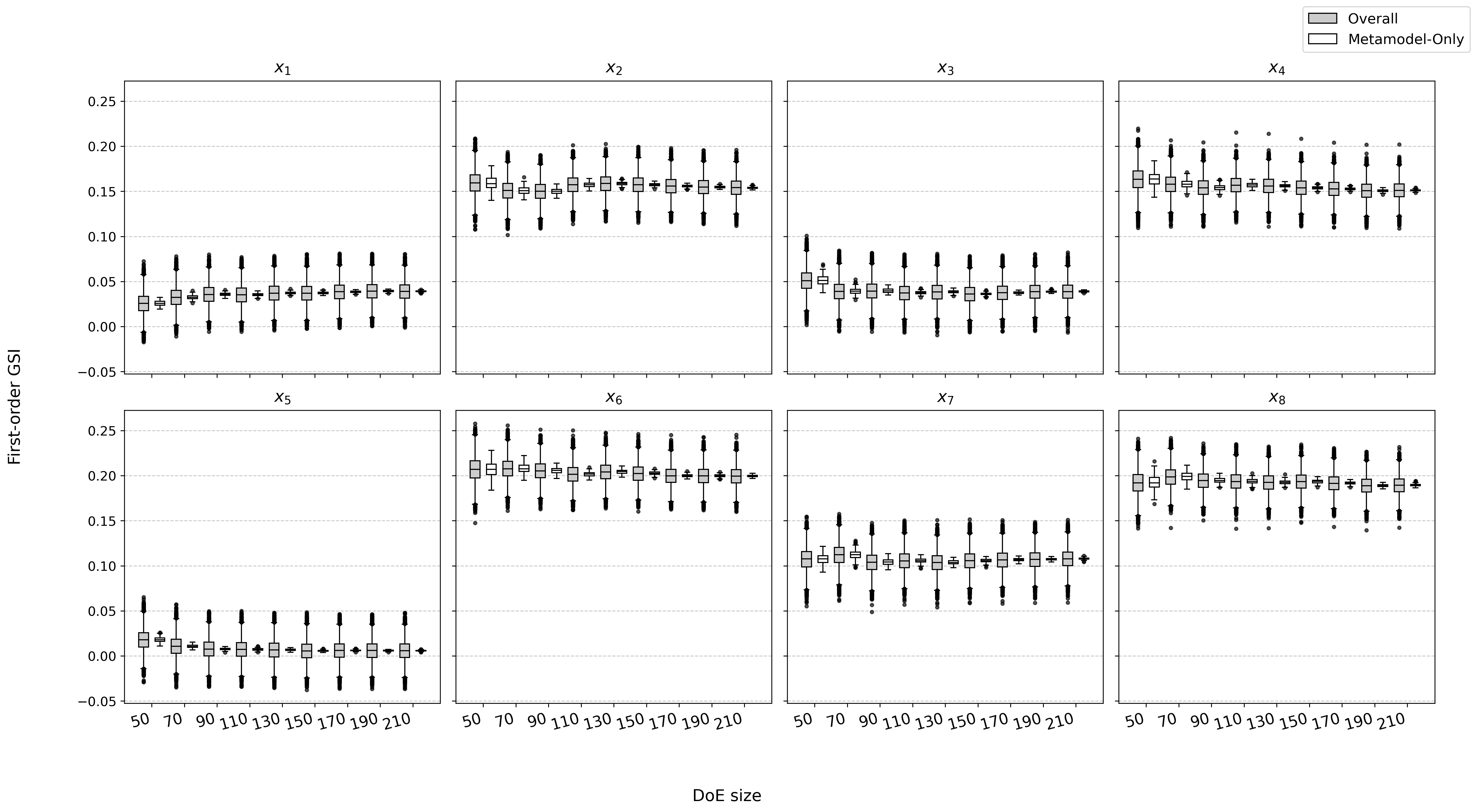}
    \caption{Campbell2D case: boxplots of first-order GSI of all input variables ($x_1,\dots, x_8$) in function of $\nDoE$. The following parameters were used: $\nbasis = 10$, $\NumberGPite=200$, $\nBoot=50$ and $\nPF = 5000$.}
    \label{fig:DOE_variation}
\end{figure}

Then, we evaluate the variation of first-order GSI data according to the size of PF input samples $\nPF$. Figure \ref{fig:PF_variation} shows the comparison of first-order GSI boxplots (overall and metamodel-only data) in function of $\nPF$ for all input variables ($x_1,\dots, x_8$). In this case, we note that overall errors decrease as $\nPF$ increases, since the estimation is more accurate with more PF sample points. Therefore, we note that the proposed algorithm provides distributions of sensitivity indices, allowing the estimation of metamodeling and estimation errors separately.

\begin{figure}[h]
    \centering
    \includegraphics[width=0.9\linewidth]{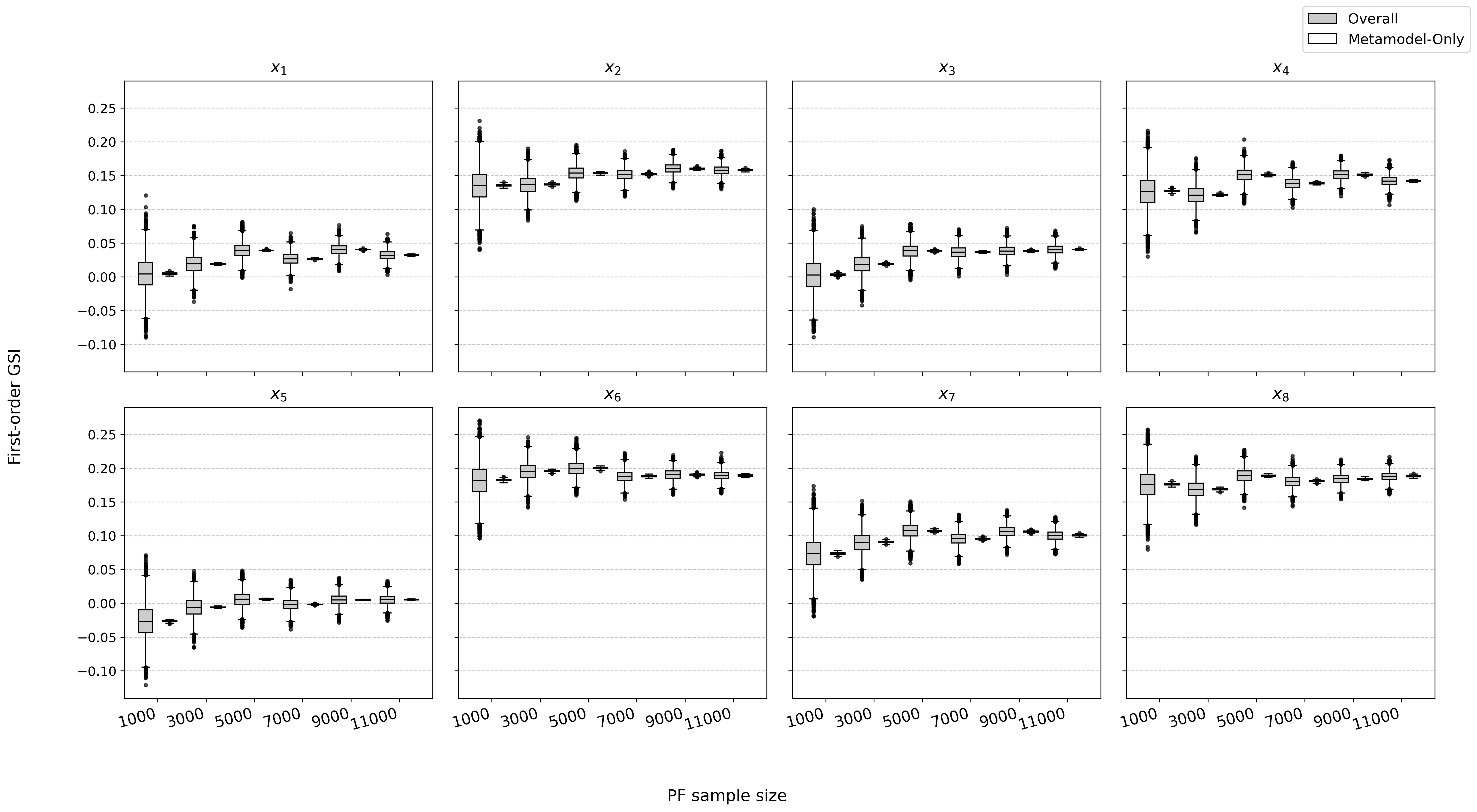}
    \caption{Campbell2D case: boxplots of first-order GSI of all input variables ($x_1,\dots, x_8$) in function of $\nPF$. The following parameters were used: $\nbasis = 10$, $\NumberGPite=200$, $\nBoot=50$ and $\nDoE = 200$.}
    \label{fig:PF_variation}
\end{figure}  

In terms of computational cost, both algorithms were tested over one dimension by considering $\nOutputDimensions=4096$, $\nPF=1000$, $\NumberGPite=10$ and $\nBoot=10$. The code run in an Apple M3 with 8 cores. Given pre-calculated GPR predictions, Algorithm \ref{algo:leGratiet} took approximately 35 seconds to complete the nested loop, whereas Algorithm \ref{proposed_algorithm} took approximately 2 seconds. It means that the Algorithm proposed in this work is about 20 times more efficient, which is similar to the observed efficiency reported in \citep{Sao2025}. For finer results, considering $\nvar=8$, $\nPF=5000$, $\NumberGPite=200$, $\nBoot=50$ and pre-calculated GPR predictions, the entire application of Algorithm \ref{proposed_algorithm} took approximately 25 minutes to complete. Therefore, the computational gain brought by the proposed algorithm is relevant, since the computational cost is reduced from several hours to dozens of minutes. On the other hand, a limitation is the fact that sampling random trajectories from the GPR is costly, because it performs Cholesky decomposition, whose complexity is $\mathcal{O}((2\nPF)^3)$, where $\nPF$ is large. \citet{legratiet2014} elaborates on efficient strategies to sample random trajectories.

The next section explores a more complex application from non-Newtonian fluid mechanics, where functional GSA is employed to study the idealized case of a gradual dam-break flow of non-Newtonian fluids. The Sobol' indices in the original output space are also evaluated, along with the $\GSI$. 

\subsection{Gradual dam-break flow of non-Newtonian fluid}
\label{subsec:application}

A gradual dam-break flow of non-Newtonian fluid consists of a known volume of non-Newtonian material inside a reservoir delimited by the walls and a gate, which is lifted with a finite velocity and the material flows downstream a horizontal plane or channel (Fig. \ref{Fig.schematic}) \citep{Ancey2009}. This application is relevant to many areas of engineering in the context of evaluating rheological properties and physical characteristics of non-Newtonian materials, such as fresh concrete, mud, gels, food, etc \cite{Pashias1996, Roussel2005, Balmforth2007, Gao2015, Pereira2022}.

\begin{figure}[h] 
\centering
\includegraphics[width=0.7\textwidth]{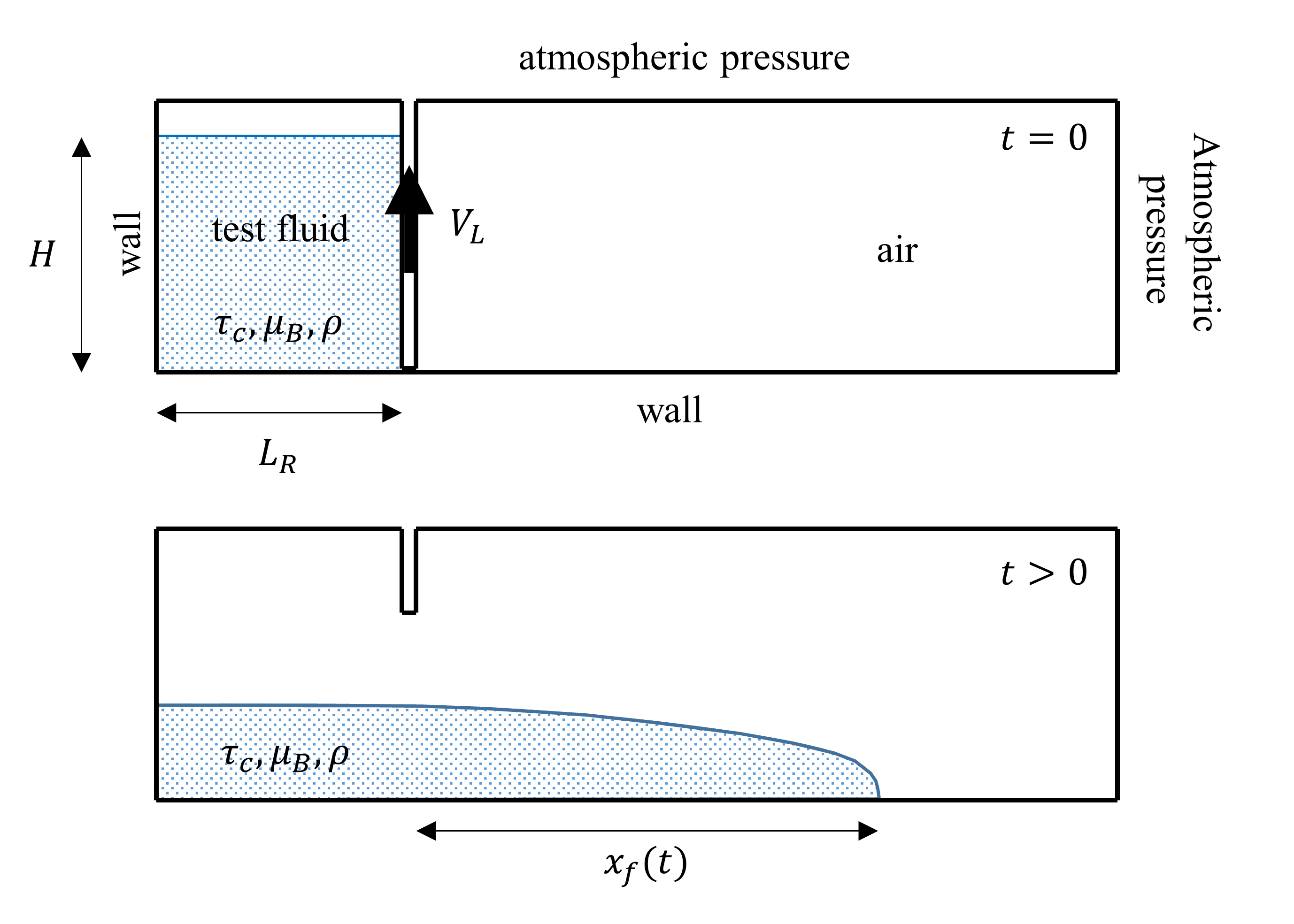}
\caption{Schematic of the idealized gradual dam-break case with boundary conditions and input variables.}\label{Fig.schematic}
\end{figure}

The settings of this case are sufficiently described in \cite{Sao2025}. Briefly, a high-fidelity deterministic code (ANSYS Fluent) is used to simulate the scenarios, where five input variables were considered to be random and uniformly distributed variables: rheological properties yield stress $\tau_c \in \mathcal{U}(0.1,\ 200.0) \ \text{Pa}$ and plastic viscosity $\mu_B \in \mathcal{U}(0.01,\ 15.00) \ \text{Pa.s}$, density $\rho \in \mathcal{U}(1000,\ 2650)$ $\text{Kg.m\textsuperscript{-3}}$, initial height $H \in \mathcal{U}(0.2, \ 1.0) \ \text{m}$ and lifting velocity $V_L \in \mathcal{U}(0.01, \ 1.00)$ $\text{m.s\textsuperscript{-1}}$. The quantity of interest is the position of the wavefront over time ($x_f \times t$) discretized in $\nOutputDimensions=5000$ output dimensions and the GSA provides sensitivity indices over the entire time series, which allows to have insights about the phenomena considering lifting dynamics, fluid characteristics and initial geometry. Due to the high-dimensionality of the output, which is a time series, functional-GSA is applied.

A pure Monte-Carlo approach was used to generate an initial DoE of size $130$. Then, the DoE was enriched with a LHS sample, producing a final sample of size $226$. The sample size was limited by the computational cost of each simulation. In this case, we used $\nbasis = 10$ principal components for all DoE subsets. A similar discussion about the accuracy of metamodels using the $Q^2$ metric can be made here (see Subsec. \ref{subsec:campbell2d_results}). Appendix \ref{app:q2} shows the $Q^2$ series over time, containing the 5th, 50th and 95th percentiles, where $200$ training points and $26$ validation points are used, with $\NumberGPite=100$ random GP trajectories. It is clear that the $Q^2$ varies according to the trajectory and, thus, metamodeling error is present because the predictions are not perfect.

Figure \ref{fig:GSI_DOE_application} shows the first-order GSI results in form of boxplots for all input variables in function of the DoE size, with a fixed PF input samples size of $\nPF=5000$. Here, we limit the presentation to first-order GSI and Sobol' indices; total indices results are available in Appendix \ref{app:total_order}. The DoE was constructed based on $50$ initial random points, which were enriched by steps of $20$ until the full DoE of size $226$. For small DoEs, the metamodel-only error dominates the overall error. We can see in Fig. \ref{fig:q2_dambreak} that the accuracy of the metamodels vary with the trajectory, which is linked to metamodeling errors. Then, as the DoE size increases, the intervals of metamodel-only boxplots decrease. This is expected and observed in Subsection \ref{subsec:campbell2d_results}, since by increasing the DoE size, the metamodels' quality increases and, therefore, there are less errors related to the metamodeling procedure. 

\begin{figure}[h]
    \centering
    \includegraphics[width=0.8\linewidth]{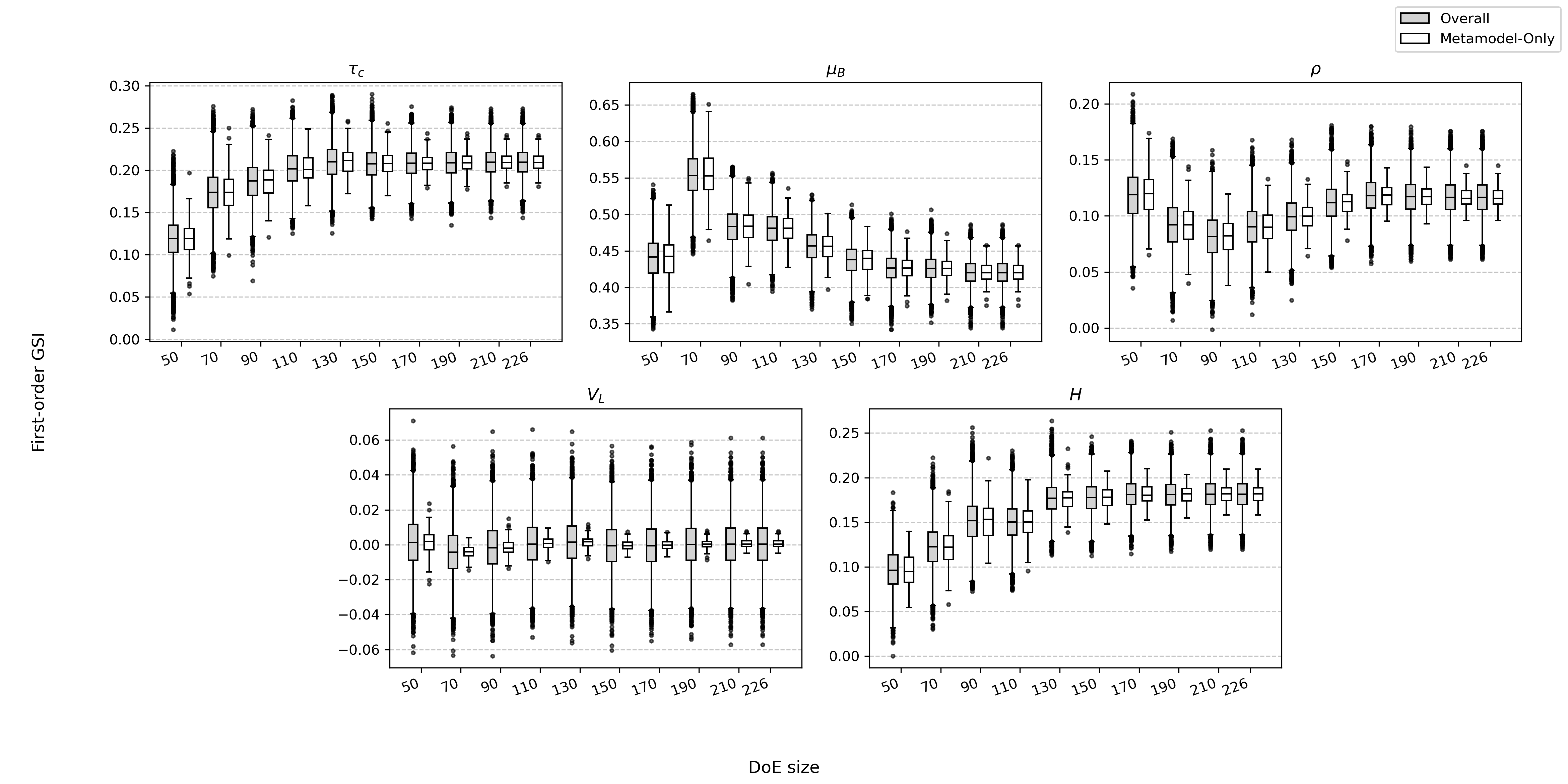}
    \caption{Dam-break case: first-order GSI boxplots of all input variables in function of DoE size $\nDoE$, with number of PF points $\nPF=5000$, number of bootstrap repetitions $\nBoot=50$ and number of GP trajectories $\NumberGPite=200$.}
    \label{fig:GSI_DOE_application}
\end{figure}

Figure \ref{fig:GSI_PF_application} shows the first-order GSI results in form of boxplots for all input variables in function of the PF input samples size, with a fixed DoE size $\nDoE=226$ and algorithm parameters $\NumberGPite=200$ and $\nBoot=50$. The PF sample size ranges from $1000$ to $11000$ points. By increasing the PF sample size, we should notice two points. First, the metamodel-only boxplots remain approximately constant even with the variation of PF points. Second, the overall boxplots' intervals decrease in function of the increasing PF sample size, since the overall error is also composed by the PF estimation error. 

\begin{figure}[h]
    \centering
    \includegraphics[width=0.8\linewidth]{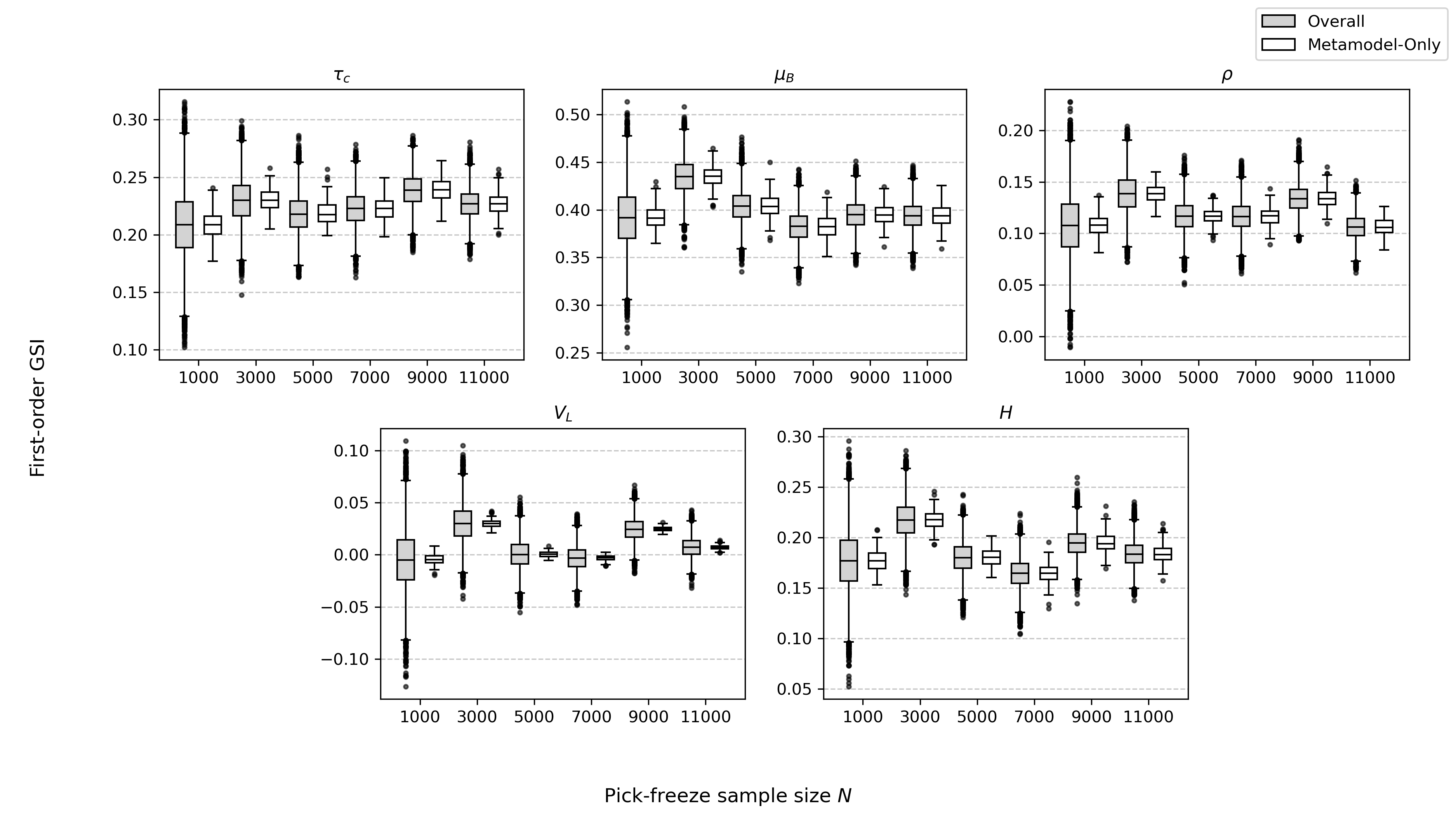}
    \caption{Dam-break case: first-order GSI boxplots of all input variables in function of the PF sample size $\nPF$, with DoE size $\nDoE=226$, number of bootstrap repetitions $\nBoot=50$ and number of GP trajectories $\NumberGPite=200$.}
    \label{fig:GSI_PF_application}
\end{figure}

Figure \ref{fig:Sobol_first_order_confidence_intervals_application} shows the first-order Sobol' indices in function of the output dimensions (time) for each input variable. The confidence intervals (metamodel-only and overall) are reported around the median values, through the 5th and 95th percentiles). A brief physical interpretation is that inertia-related input variables, such as the lifting velocity $V_L$ and initial height $H$ (related to the initial stored potential energy), are more important in the beginning of the flow. As the flow develops over time, the pressure driving force diminishes and resistance-related parameters become the most influential terms, i.e. the rheological parameters yield stress $\tau_c$ and plastic viscosity $\mu_B$.

\begin{figure}[h]
    \centering
    \includegraphics[width=0.8\linewidth]{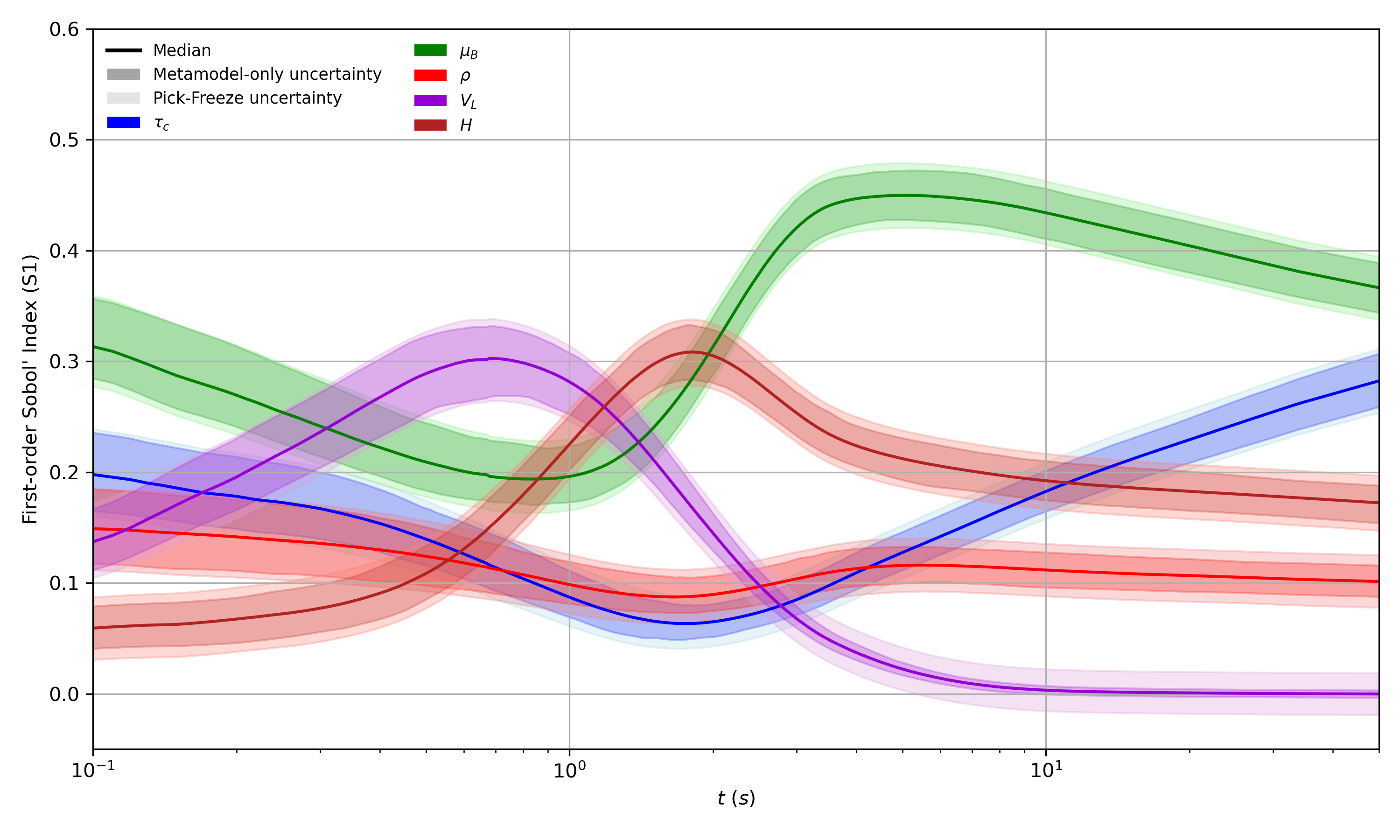}
    \caption{Dam-break case: median, 5th and 95th percentiles of first-order Sobol' indices over output dimensions (time). Fixed parameters: DoE size $\nDoE = 226$, PF sample size of $\nPF=11000$, number of bootstrap repetitions $\nBoot=50$ and number of GP trajectories $\NumberGPite=200$.}
\label{fig:Sobol_first_order_confidence_intervals_application}
\end{figure}

The medians show a good agreement with the kriging mean prediction, as shown in \cite{Sao2025}. The confidence intervals are separated into metamodel-related and PF-estimation errors, where the metamodeling is the main source of error. In fact, Fig. \ref{fig:relative_contribution_application} shows the contribution of both error sources over the output dimensions and confirms this observation, with exception to $V_L$ indices. For later times, the PF estimation turns to be the main source of error, since the indices are close to zero. A different estimator other than the Janon-Monod estimator may be more suitable for this case, however this is not the goal of this analysis. 

\begin{figure}[h]
    \centering
    \includegraphics[width=0.8\linewidth]{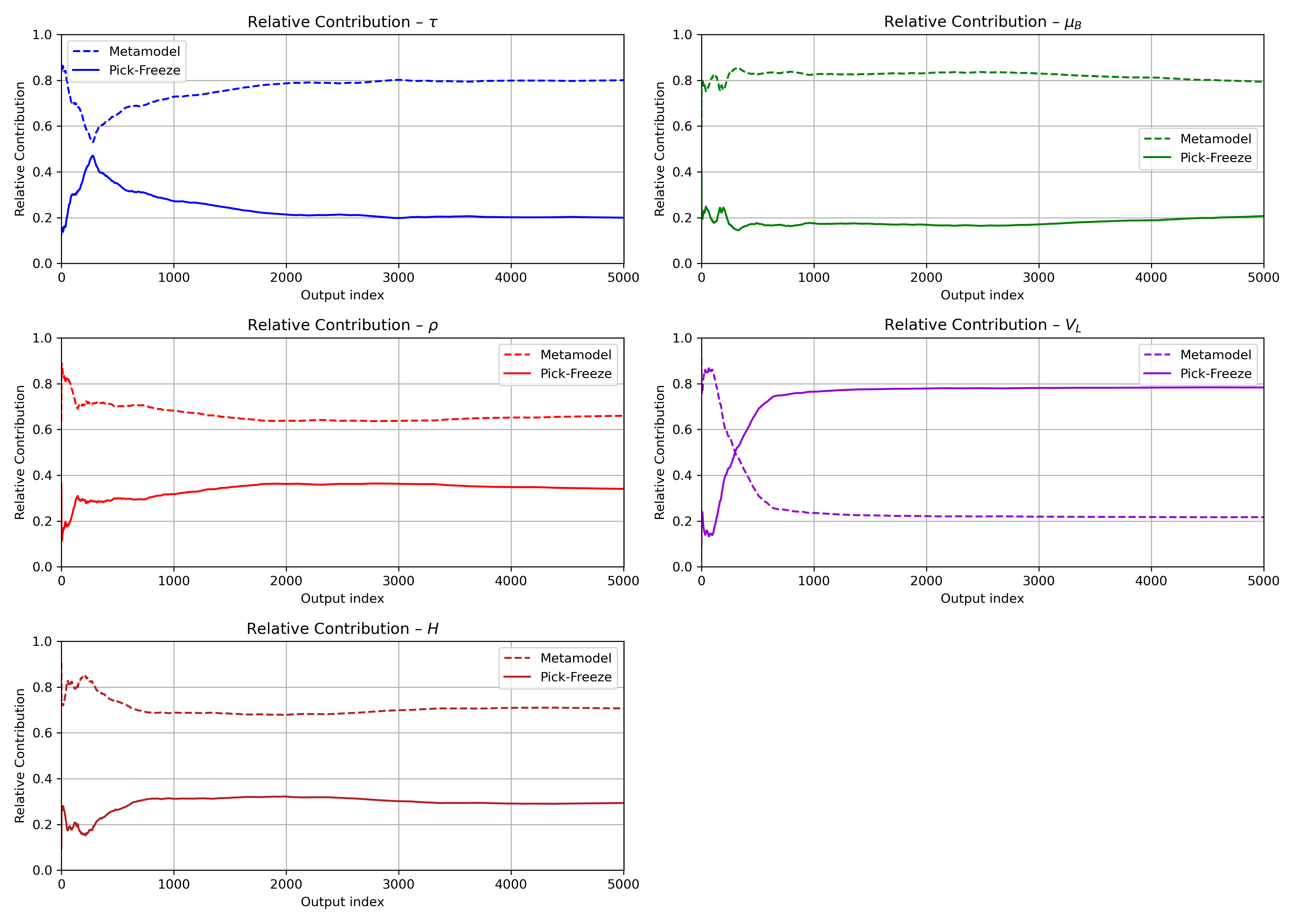}
    \caption{Dam-break case: relative contribution of metamodeling-related errors ($\nDoE=226$) and PF estimation errors ($\nPF=11000$) over output dimensions (time).}
    \label{fig:relative_contribution_application}
\end{figure}

This behavior is expected, given the first-order GSI results in Fig. \ref{fig:GSI_PF_application}, where the metamodel-only boxplot ranges are comparable to the overall boxplot ranges. Therefore, in this case, the PF estimation error is low due to the size of PF input samples ($\nPF=11000$). 

In terms of computational cost, the same test conducted with Campbell2D function (Subsection \ref{subsec:campbell2d_results}) was performed here, by considering $\nOutputDimensions = 5001$, $\nPF=1000$, $\NumberGPite=10$ and $\nBoot=10$. The code run in an Apple M3 with 8 cores. Given pre-calculated GPR predictions, Algorithm \ref{algo:leGratiet} took approximately 30 seconds to complete the nested loop, whereas Algorithm \ref{proposed_algorithm} took approximately 2 seconds. The results were similar to the reported in Subsection (\ref{subsec:campbell2d_results}). For finer results, considering $\nvar=5$, $\nPF=11000$, $\NumberGPite=200$, $\nBoot=50$ and pre-calculated GPR predictions, the entire application of Algorithm \ref{proposed_algorithm} took approximately 32 minutes to complete. As observed in Subsection \ref{subsec:campbell2d_results}, it is clear that the novel algorithm allows a faster computation, reducing the computational time from several hours to dozens of minutes.

\section{Conclusion}
\label{sec:conclusion}

In this work, we introduce a computationally efficient algorithm capable of producing distributions of sensitivity indices in a functional GSA context. These distributions allow the representation of errors inherent to the functional GSA procedure, such as metamodeling and estimation errors. Here, specifically, we use the Gaussian process regression (GPR) technique since it can predict multiple trajectories for the same input sample.

To this end, the algorithm requires functional data to be expanded into a functional basis and metamodeling of the few truncated basis coefficients of the expansion. Then, for a single GPR trajectory, a vector-valued \textit{pick-freeze} (PF) estimator is used to estimate the Sobol' indices and generalized sensitivity index (GSI). A bootstrap strategy is applied in the PF estimation of the sensitivity indices, which is computationally efficient. This procedure is applied to multiple GP trajectories and two distributions of sensitivity indices can be obtained: one related only to the metamodeling error and one related to the overall errors. This framework is an adaptation of \citet{legratiet2014} algorithm, originally developed for scalar outputs, for functional GSA.

The framework was applied to two test cases: the analytical Campbell2D function and the data-driven case of a gradual and idealized dam-break flow of non-Newtonian fluid. The errors were visualized by using boxplots and the isolated effect of the DoE size and the PF sample sizes were examinated. For a fixed PF sample size, the metamodeling-only error reduces in function of the DoE's enrichment. In a similar way, for a fixed DoE size, the overall error reduces in function of the PF sample size, while the metamodeling-only error remained approximately constant. The results showed that the methodology was able to provide larger errors for poorer DoEs and for smaller PF sample sizes, which is the expected behavior since these conditions are prone to more error. 
Finally, in terms of computational efficiency, the proposed algorithm is more efficient than applying the \citet{legratiet2014} algorithm separately for each output dimension. The reason is that the number of basis components is significantly smaller than the number of output dimensions, resulting in fewer mathematical operations \citep{Sao2025}. The results on computational performance show that the proposed algorithm is about $15$ times faster.

We will mention two perspectives for this work. 
First, the proposed algorithm automatically gives the (joint) distribution of the \textit{vector} of Sobol' indices estimators at all pixels of the sensitivity map. Although only the marginal distribution has been illustrated here, this property could be used to assess the model error of an estimator involving all the pixels (or time-steps) such as the maximum of the Sobol' indices PF estimators over the space or time period. 
Secondly, the algorithm was implemented in the standard setting where each principal component is modeled independently by a GP. It can be extended to more advanced models where the vector of principal components is modeled jointly by a multioutput GP (see e.g. \cite{Alvarez_Lawrence_kernelsVectorValued}), since the algorithm only requires a simulation procedure for conditional paths. 

\bibliographystyle{abbrvnat}   
\bibliography{cas-refs}
\appendix

\section{Analysis of Remark \ref{remark:covariance_estimation}, number of GP trajectories and bootstrap repetitions}\label{app:nGPR_nBoots}

Regarding Remark \ref{remark:covariance_estimation}, the overall covariance can be calculated by two approaches: (i) by considering a fixed output data from the DoE (fixed covariance); or (ii) by computing the empirical overall covariance at each trajectory and at each bootstrap repetition (empirical covariance). Fig. \ref{fig:covariance_estimations_GSI} shows that both approaches show negligible differences for the evaluation of first-order GSI. This subsection also shows that the number of GP trajectories and of bootstrap repetitions taken by both test-cases $(\NumberGPite, \nBoot)=(200, 50)$ is sufficient. Figures \ref{fig:number_boots_analysis_campbell} and \ref{fig:number_GPR_trajectories_analysis_campbell} refer to the Campbell2D function (Subsection \ref{subsec:campbell2d_results}) and Figs. \ref{fig:number_boots_analysis_practical} and \ref{fig:number_GPR_trajectories_analysis_practical} refer to the practical case (Subsection \ref{subsec:application}).

\begin{figure}[H]
    \centering    \includegraphics[width=0.8\linewidth]{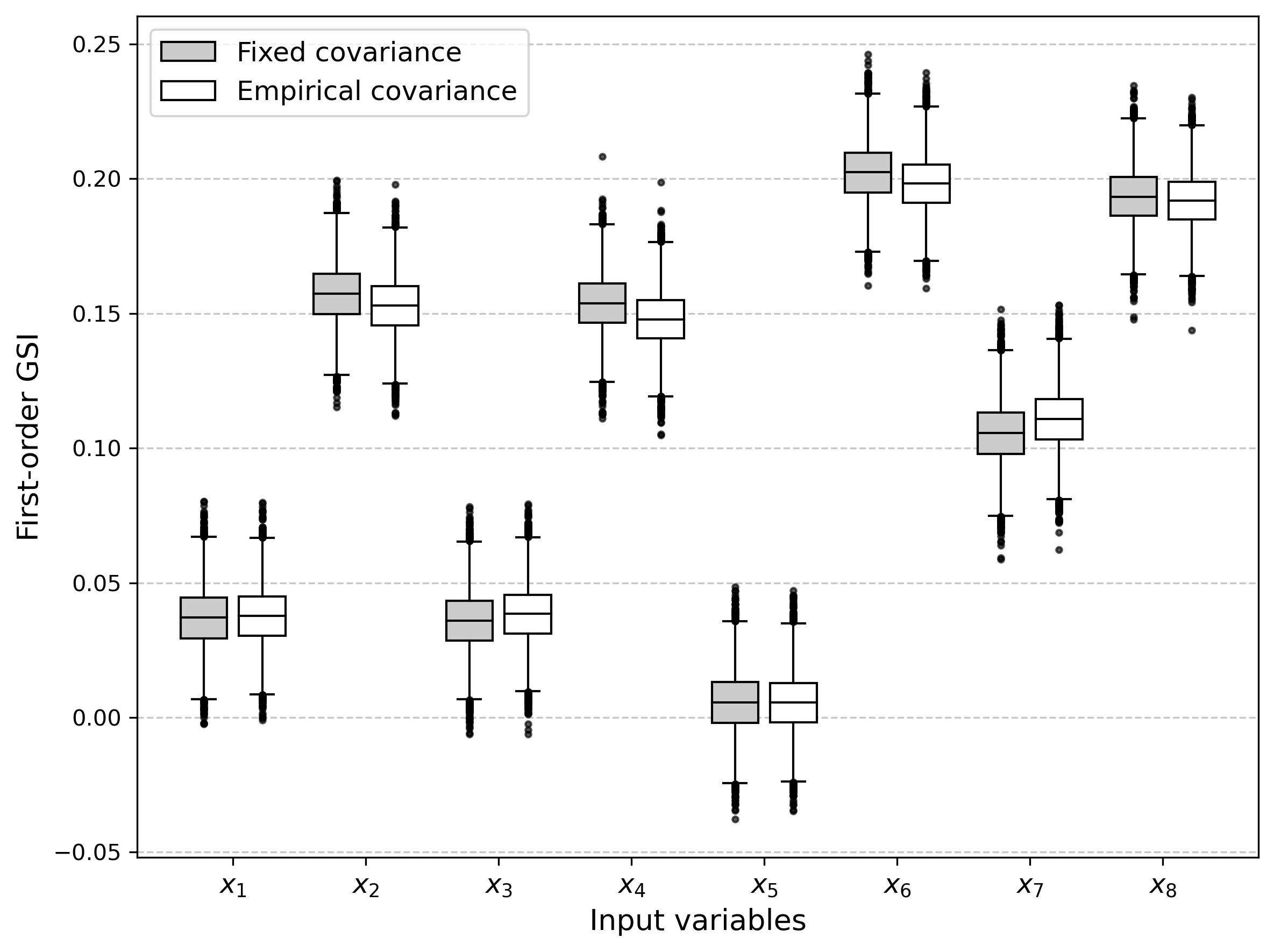}
    \caption{Campbell2D case: comparison between two approaches to compute overall covariance (Remark \ref{remark:covariance_estimation}), considering first-order GSI of all input variables. In the fixed covariance estimation, the overall covariance is taken from DoE output data; in the empirical overall covariance estimation (Eq. \ref{eq:gsi_estimator}), the overall covariance is computed with the results from the PF input samples. We consider $\nDoE = 150$, $\nPF=5000$, $\nBoot=50$ and $\NumberGPite=200$.}
\label{fig:covariance_estimations_GSI}
\end{figure}

\begin{figure}[H]
  \centering
  \begin{subfigure}[t]{0.48\textwidth}
    \centering
    \includegraphics[width=\linewidth]{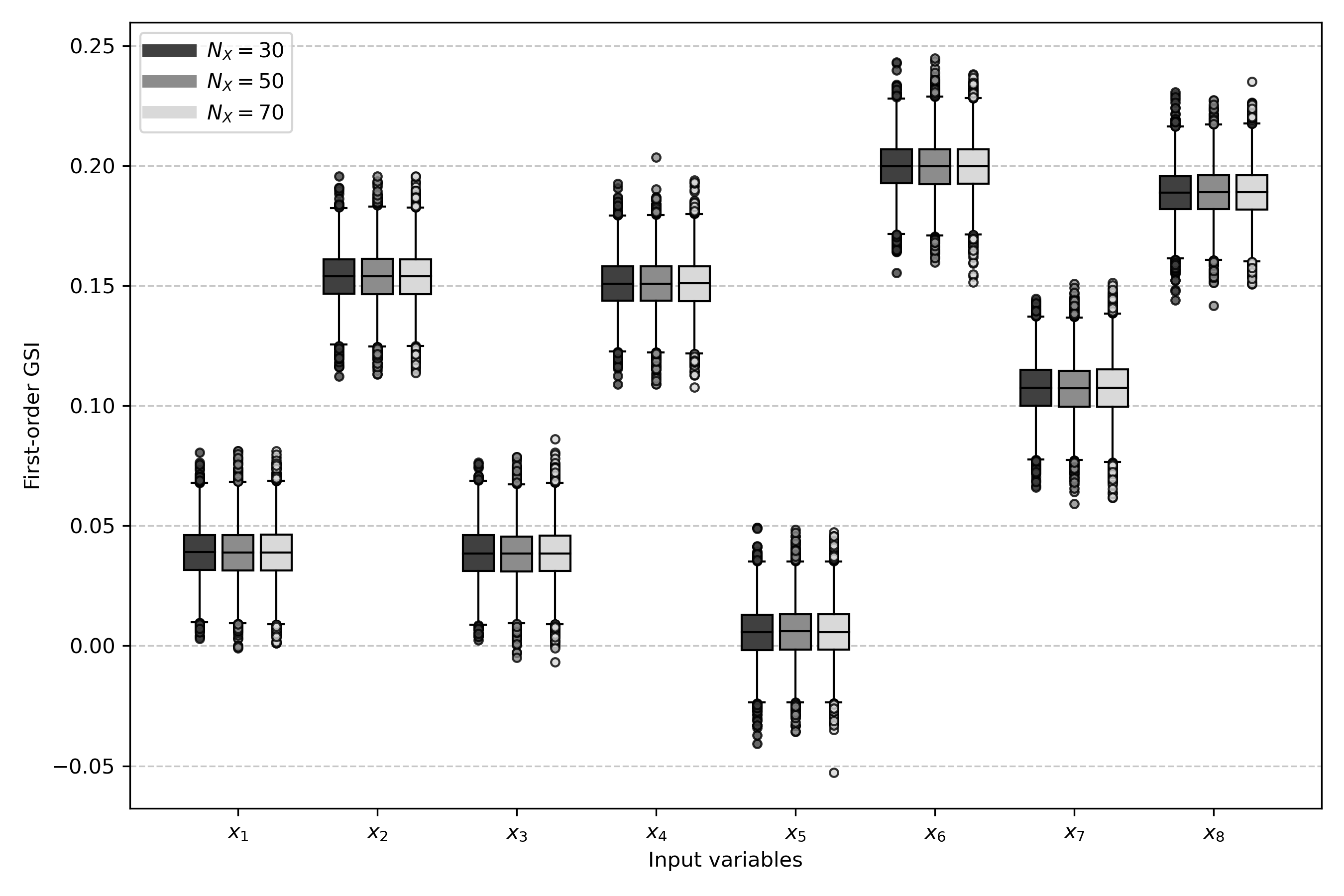}
    \caption{Comparison between bootstrap iterations $\nBoot=(30, 50, 70)$ using the first-order GSI of each input variable. Size of PF input samples $\nPF=5000$ and number of GP trajectories $\NumberGPite=200$ were fixed.}
    \label{fig:number_boots_analysis_campbell}
  \end{subfigure}
  \hfill
  \begin{subfigure}[t]{0.48\textwidth}
    \centering
    \includegraphics[width=\linewidth]{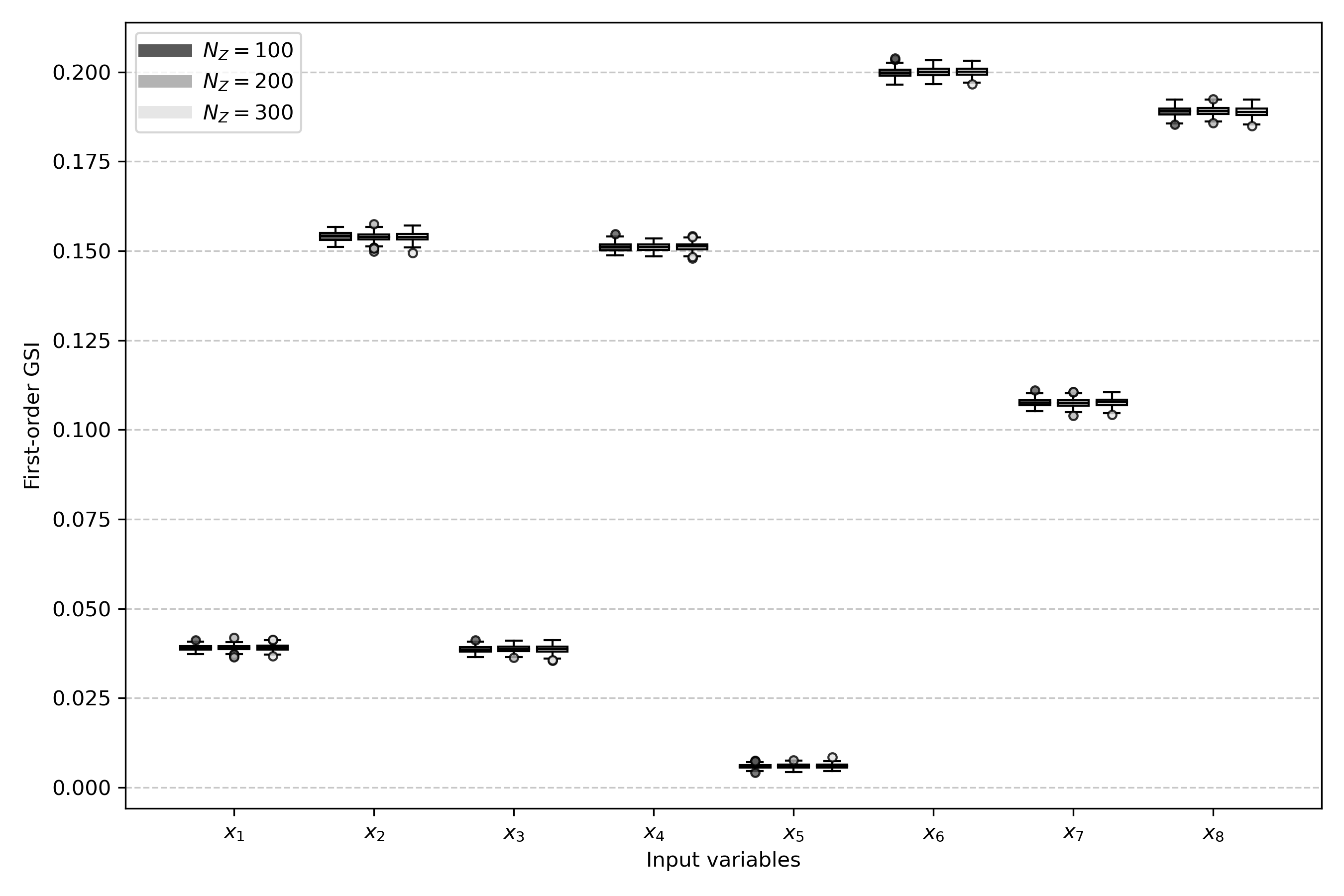}
    \caption{Comparison between numbers of GP trajectories $\NumberGPite=(100, 200, 300)$ using the first-order GSI of each input variable. Size of PF input samples $\nPF=5000$ and number of bootstrap repetitions $\nBoot=50$ were fixed.}
    \label{fig:number_GPR_trajectories_analysis_campbell}
  \end{subfigure}
  \caption{Analysis of the number of GP trajectories and bootstrap repetitions (Campbell2D function).}
\end{figure}

\begin{figure}[htbp]
  \centering
  \begin{subfigure}[t]{0.48\textwidth}
    \centering
    \includegraphics[width=\linewidth]{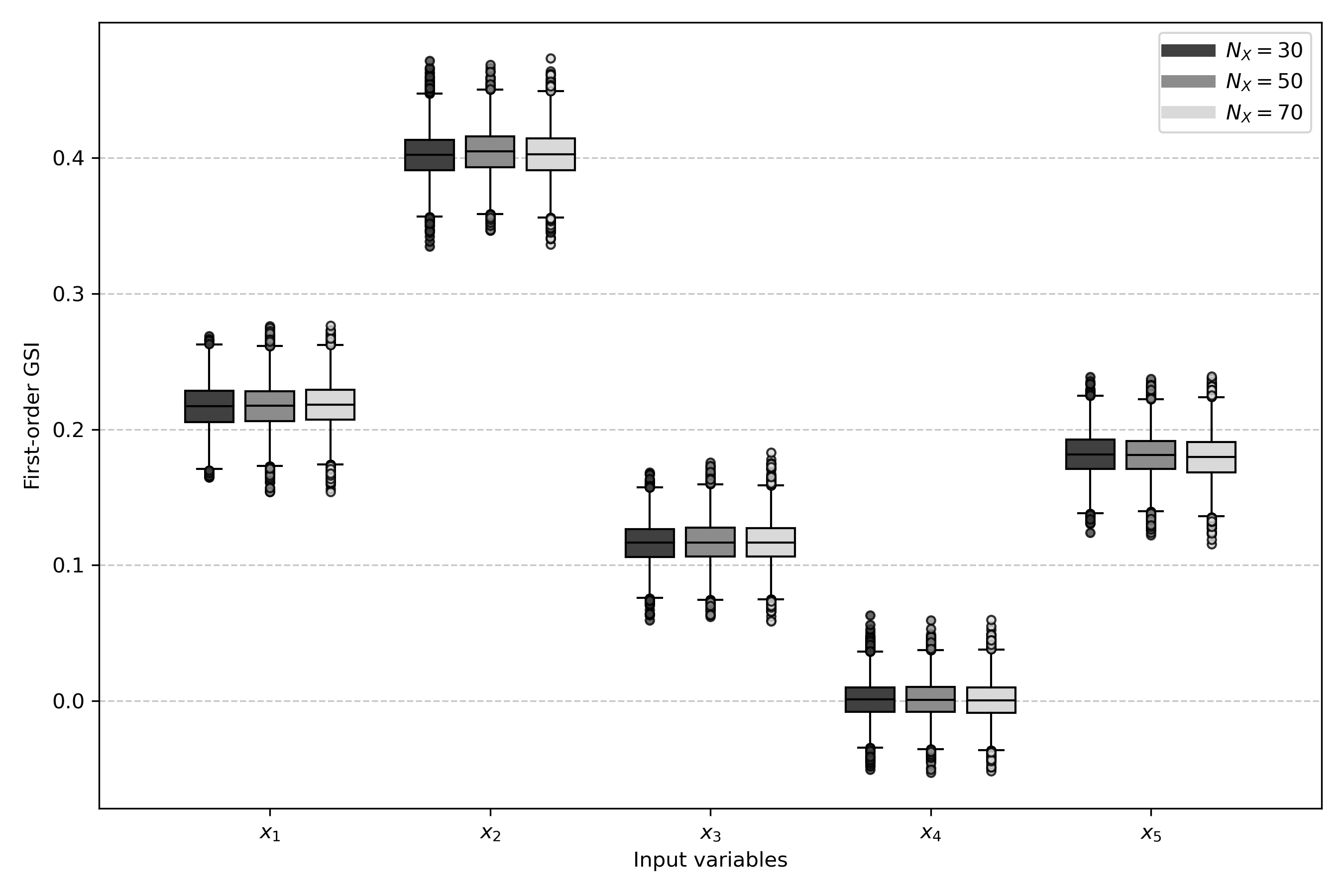}
    \caption{Comparison between number of bootstrap iterations $\nBoot=(30, 50, 70)$ using the first-order GSI of each input variable. Size of PF input samples $\nPF=5000$ and number of GP trajectories $\NumberGPite=200$ were fixed.}
    \label{fig:number_boots_analysis_practical}
  \end{subfigure}
  \hfill
  \begin{subfigure}[t]{0.48\textwidth}
    \centering
    \includegraphics[width=\linewidth]{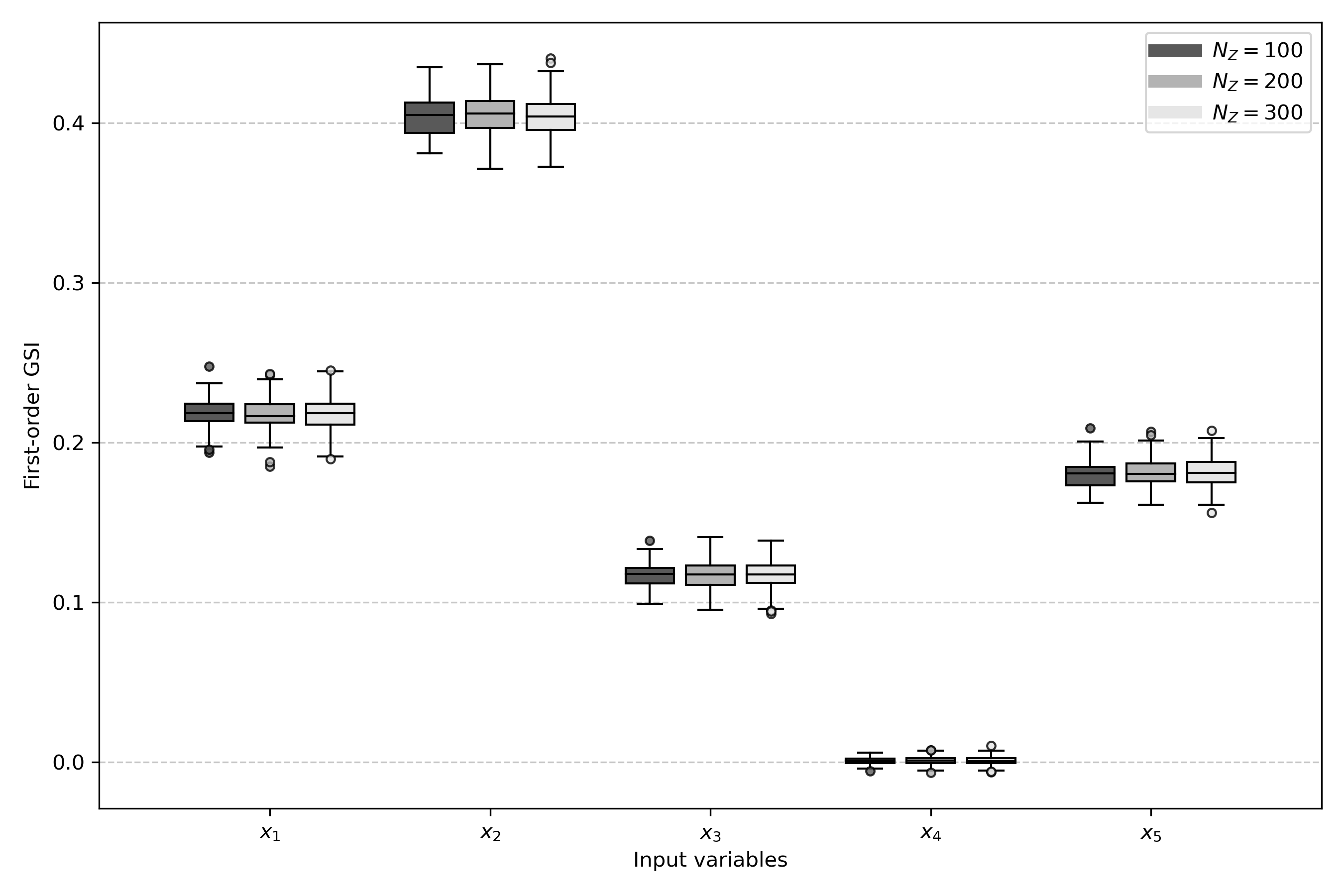}
    \caption{Comparison between number of GP trajectories $\NumberGPite=(100, 200, 300)$ using the first-order GSI of each input variable. Size of PF input samples $\nPF=5000$ and number of bootstrap repetitions $\nBoot=50$ were fixed.}
    \label{fig:number_GPR_trajectories_analysis_practical}
  \end{subfigure}
  \caption{Analysis of the number of GP trajectories and bootstrap repetitions (gradual dam-break of non-Newtonian fluid flow case).}
\end{figure}

\section{Total indices results}\label{app:total_order}

Algorithm \ref{proposed_algorithm} is adaptable to evaluate higher-order indices. This appendix shows results of the total GSI (for both cases) and total Sobol' indices at all output dimensions for the dam-break application, described in Subsection \ref{subsec:application}. 

\subsection{Campbell 2D function}

Considering the Campbell2D function, Figures \ref{fig:total_GSI_campbell_DOE} and \ref{fig:total_GSI_campbell_PF} show the results of total GSI for different DoE subset sizes and PF sample sizes, respectively.

\begin{figure}[h]
    \centering    \includegraphics[width=0.8\linewidth]{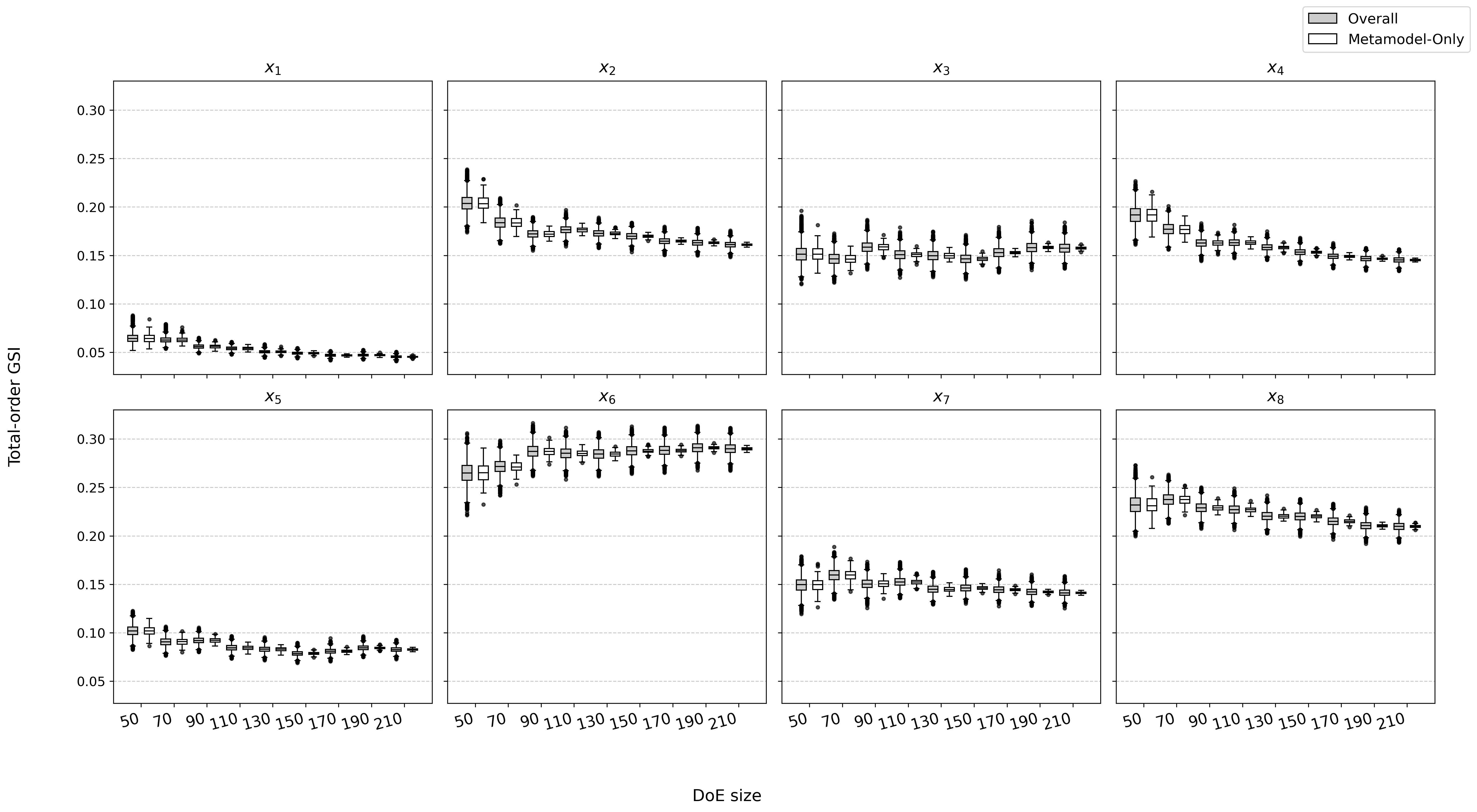}
    \caption{Campbell2D case: total GSI of all input variables in function of DoE size $\nDoE$, for $\nPF=5000$, $\nBoot=50$ and $\NumberGPite=200$.}
\label{fig:total_GSI_campbell_DOE}
\end{figure}

\begin{figure}[h]
    \centering    \includegraphics[width=0.8\linewidth]{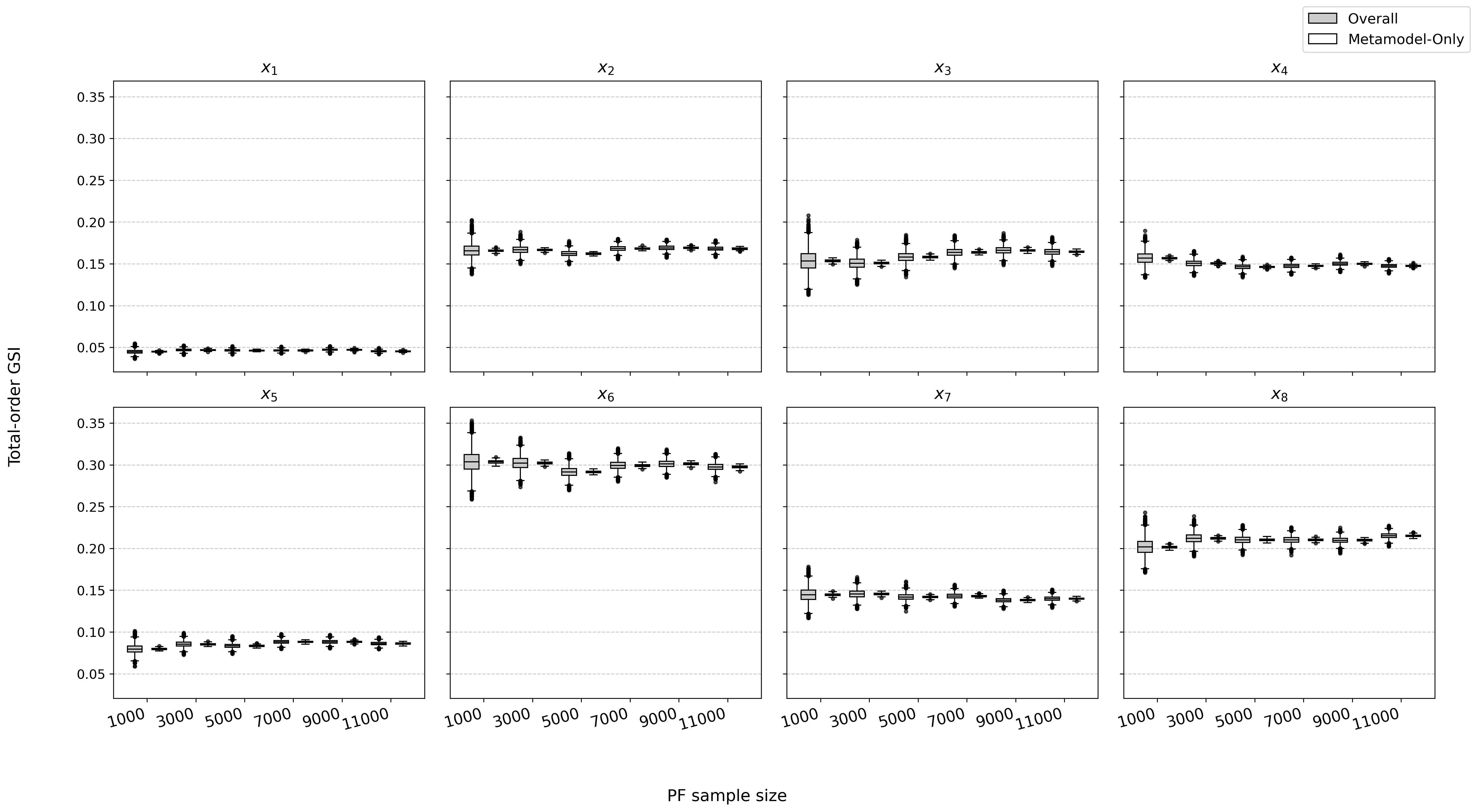}
    \caption{Campbell2D case: total GSI of all input variables in function of PF sample size $\nPF$, for $\nDoE = 200$, $\nBoot=50$ and $\NumberGPite=200$.}
\label{fig:total_GSI_campbell_PF}
\end{figure}

\subsection{Dam-break application}

Considering the dam-break application, Figures \ref{fig:total_DOE_dambreak} and \ref{fig:total_PF_dambreak} show the results of total GSI for different DoE subset sizes and PF sample sizes, respectively. Figure \ref{fig:total_sobol_dambreak} shows the total Sobol' indices reprojected to the original output dimensions, i.e. over time.

\begin{figure}[h]
  \centering
  \begin{subfigure}[t]{0.48\textwidth}
    \centering
    \includegraphics[width=\linewidth]{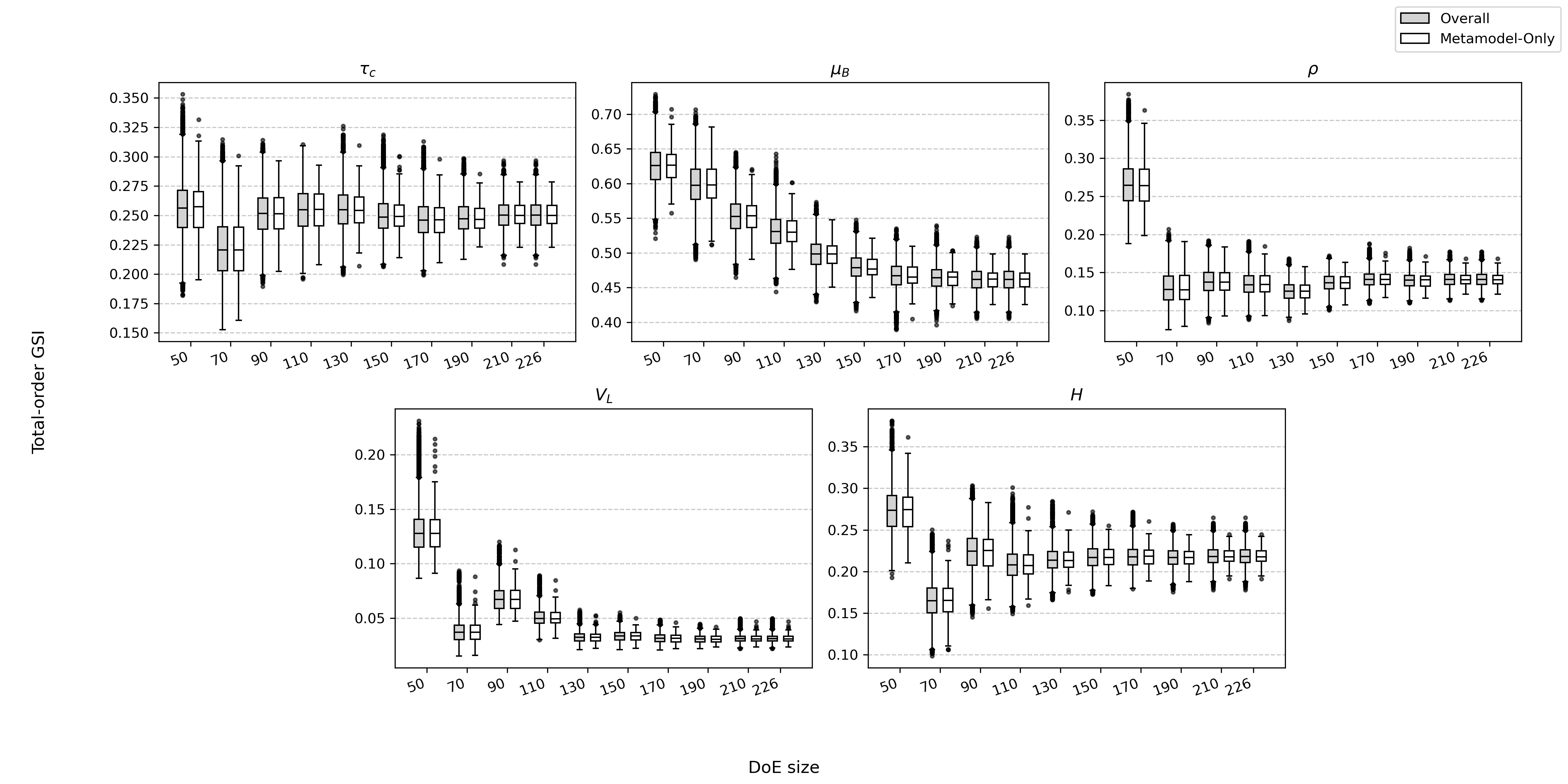}
    \caption{Total GSI results of all input variables in function of DoE size $\nDoE$, for $\nPF=5000$, $\nBoot=50$ and $\NumberGPite=200$.}
    \label{fig:total_DOE_dambreak}
  \end{subfigure}
  \hfill
  \begin{subfigure}[t]{0.48\textwidth}
    \centering
    \includegraphics[width=\linewidth]{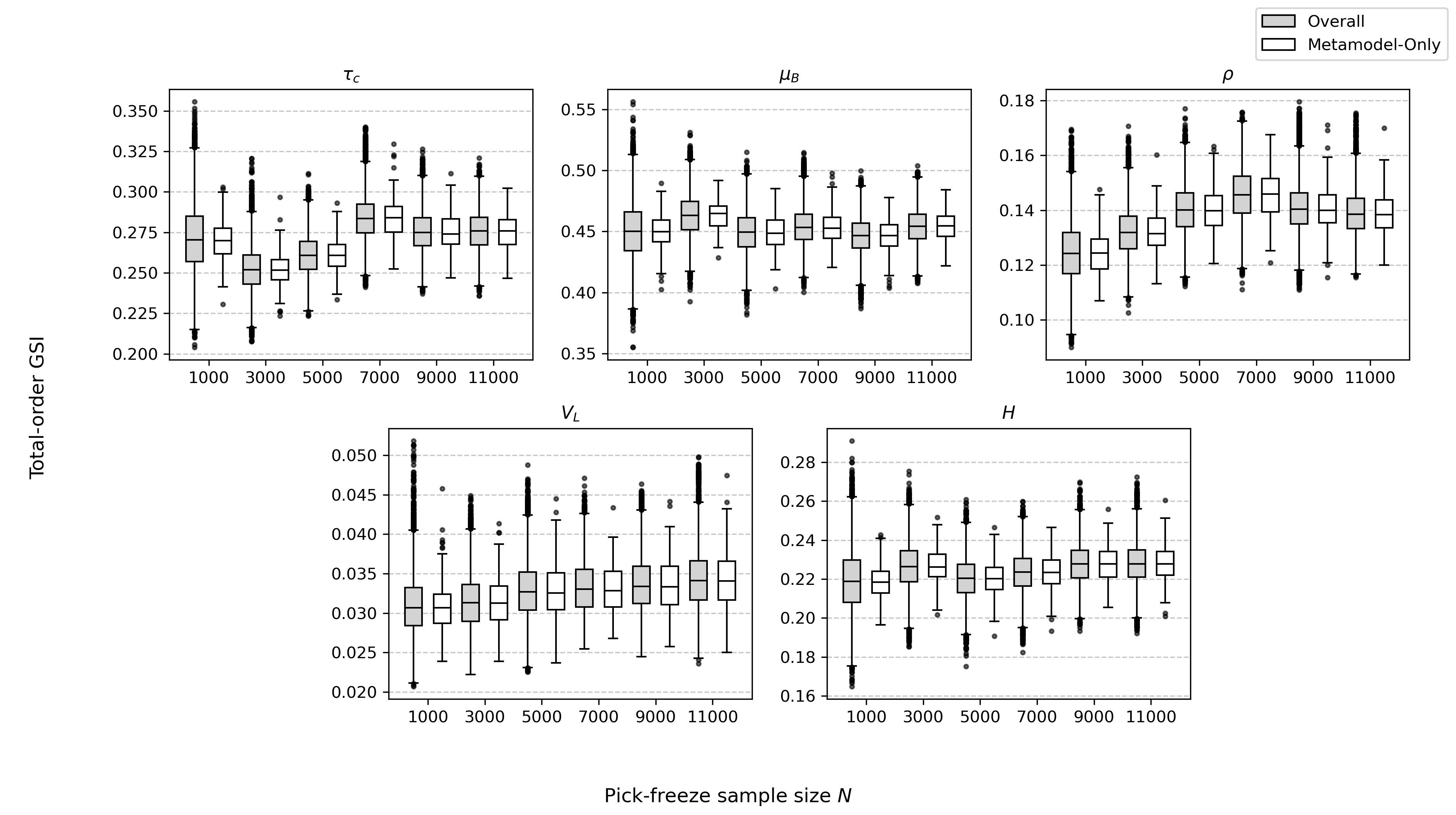}
    \caption{Total GSI results of all input variables in function of the PF sample size $\nPF$, for $\nDoE=226$, $\nBoot=50$ and $\NumberGPite=200$.}
    \label{fig:total_PF_dambreak}
  \end{subfigure}
  \caption{Boxplots of total GSI for different sizes of DoE and PF sample (Dam-break case).}
\end{figure}

\begin{figure}[h]
    \centering
    \includegraphics[width=0.8\linewidth]{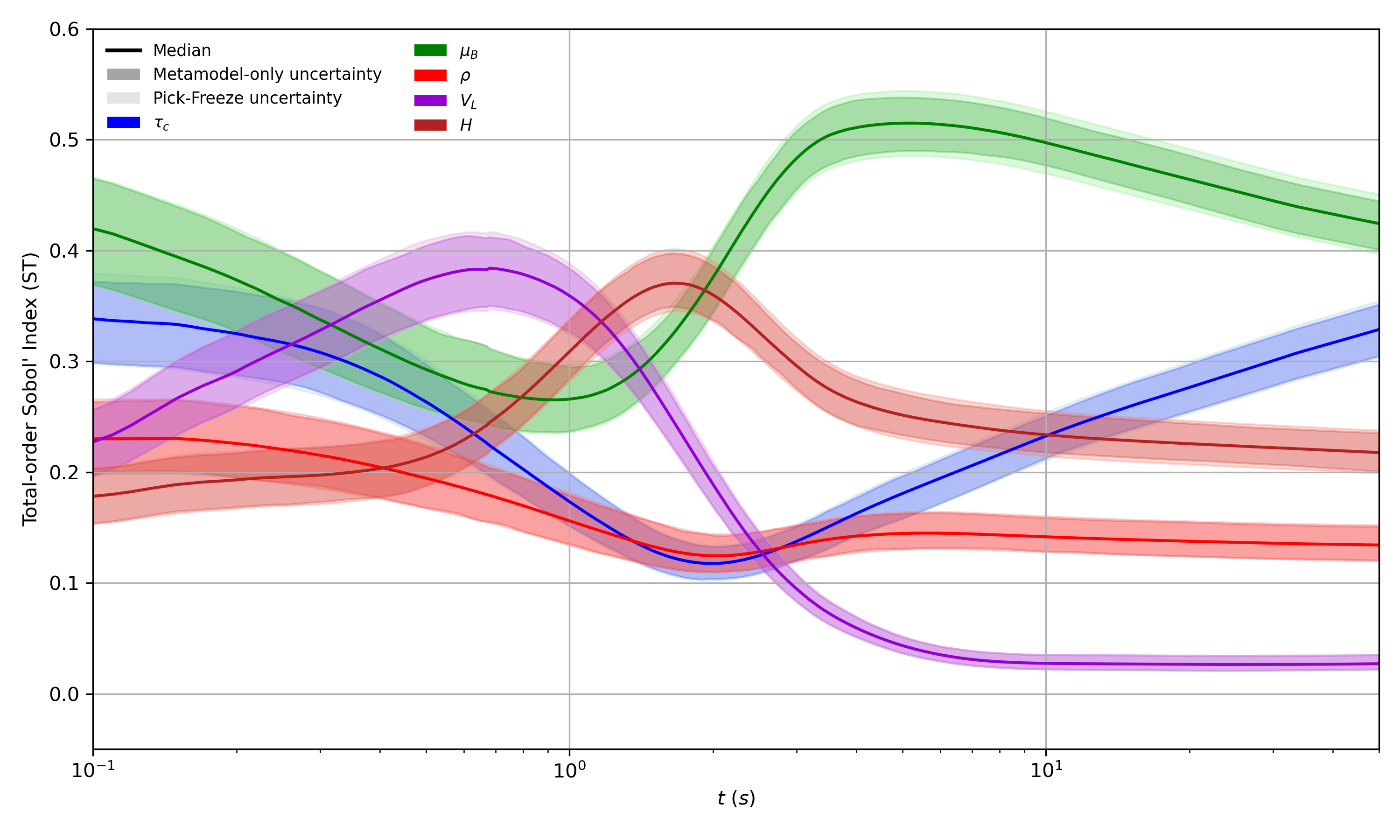}
    \caption{Dam-break case: total Sobol' indices over output dimensions (time), with DoE size $\nDoE = 226$, PF sample size of $\nPF=11000$, number of bootstrap repetitions $\nBoot=50$ and number of GP trajectories $\NumberGPite=200$.}
\label{fig:total_sobol_dambreak}
\end{figure}

\section{Validation of metamodels}
\label{app:q2}

This appendix shows the Nash-Sutcliffe efficiency (or $Q^2$ metric) for the Campbell2D function (Fig. \ref{fig:q2_campbell}) and for the dam-break application (Fig. \ref{fig:q2_dambreak}).

\begin{figure}[h]
    \centering
    \includegraphics[width=0.8\linewidth]{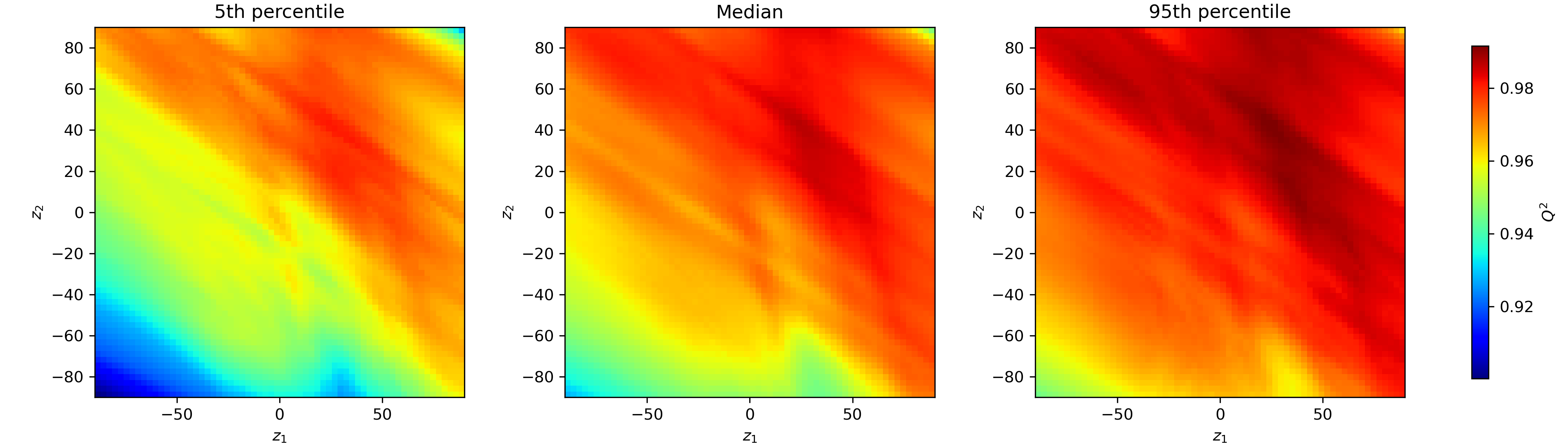}
    \caption{Campbell2D case: $Q^2$ maps and the 5th, 50th and 95th percentiles considering a set of $\NumberGPite=100$ random GP trajectories, $200$ training points and $50$ validation points.}
    \label{fig:q2_campbell}
\end{figure}

\begin{figure}[h]
    \centering
    \includegraphics[width=0.8\linewidth]{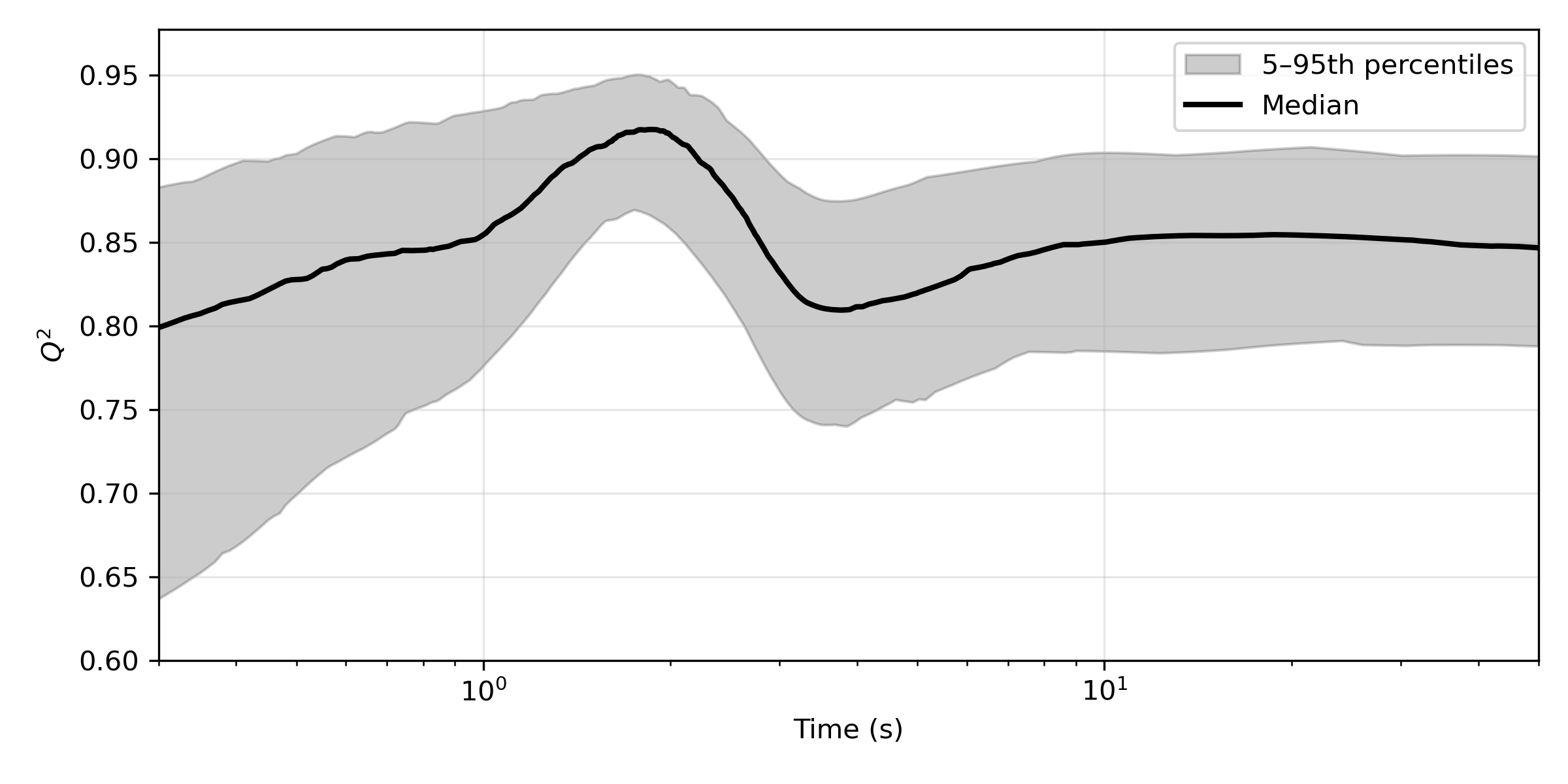}
    \caption{Dam-break case: $Q^2$ series over time and the 5th, 50th and 95th percentiles considering a set of $\NumberGPite=100$ random GP trajectories, $200$ training points and $26$ validation points}
    \label{fig:q2_dambreak}
\end{figure}

\end{document}